\documentclass[12pt]{iopart}

\usepackage{graphicx}
\usepackage{amsfonts}
\usepackage{amsmath}
\usepackage{amsthm}
\usepackage{xcolor} 
\usepackage{soul} 

\newtheorem{thm}{Theorem}[section]
\newtheorem{prop}{Proposition}[section]
\newtheorem{lem}{Lemma}[section]

\newtheorem{defi}{Definition}[section]

\begin{document}
\title[Multipartite entanglement in qudit hypergraph states]{Multipartite entanglement in qudit hypergraph states}
\author{D Malpetti$^1$, A Bellisario$^2$,
C Macchiavello$^3$}

\address{$^1$ Dalle Molle Institute for Artificial Intelligence (IDSIA), Via la Santa 1, 6962 Lugano, Switzerland}
\address{$^2$ Laboratory of Molecular Biophysics, Department of Cell and Molecular Biology, Uppsala University, Husargatan 3 (Box 596), SE-751 24 Uppsala, Sweden}
\address{$^3$ Dipartimento di Fisica and INFN-Sezione di Pavia, Via Bassi 6, 27100 Pavia, Italy}
\ead{chiara.macchiavello@unipv.it}

\begin{abstract}
We study entanglement properties of hypergraph states in arbitrary finite dimension. We compute multipartite entanglement of elementary qudit hypergraph states, namely those endowed with a single maximum-cardinality hyperedge. We show that, analogously to the qubit case, also for arbitrary dimension there exists a lower bound for multipartite entanglement of connected qudit hypergraph states; this is given by the multipartite entanglement of an equal-dimension elementary hypergraph state featuring the same number of qudits as the largest-cardinality hyperedge. We highlight interesting differences between prime and non-prime dimension in the entanglement features.
\end{abstract}
\medskip
\medskip
\medskip
\medskip
\maketitle

\section{Introduction}

Quantum hypergraph states \cite{rossi2013,qu2013} are a class of multipartite states that generalise the class of graph states (for a complete review on graph states see \cite{hein2006}). 
It was recently shown that quantum hypergraph states play a major role in various aspects of quantum information processing and in the foundations of quantum mechanics. For instance, they are employed in many quantum algorithms \cite{bruss2011}, they exhibit interesting nonlocal features \cite{guhne2014} and provide extreme violation of local realism \cite{budroni2016}, and they are a key ingredient in recent proposals of quantum artificial neural networks \cite{Tacchino2019}.

The notion of quantum hypergraph states was recently generalised to systems of arbitrary finite dimension (qudits) \cite{Guehne_2017}, and interesting features were derived in terms of entanglement classes.
Multipartite entanglement is of high interest in multipartite systems, since it is a fundamental resource in several quantum information tasks, such as secret sharing \cite{hillery1999}, multipartite quantum key distribution \cite{multi-qkd} and distributed dense coding \cite{dist-dense}. Multipartite entanglement properties of qubit hypergraph states were recently studied \cite{Ghio_2017}.
In this paper we investigate multipartite entanglement properties for qudit hypergraph states.

The paper is structured as follows. In the first two sections we review previous results in the domain: Sect.\,\ref{sect:intro_hs} is devoted to qudit hypergraph states, and Sect.\,\ref{sect:intro_ent} to multipartite entanglement. In Sect.\,\ref{sect:pi_sets} we develop a mathematical formalism that we use throughout our work.  In Sect.\,\ref{sect:ent_formula} we compute multipartite entanglement of qudit elementary hypergraph states. In Sect.\,\ref{sect:lower_bound} we prove the existence of a lower bound for multipartite entanglement of connected qudit hypergraph states.  We conclude the paper with a summary and outlook in Sect. \ref{s:Conclusions}. The technical derivations of the main results presented in the paper are reported in the Appendices.

\section{Qudit hypergraph states} \label{sect:intro_hs}

In this Section we review elements of the theory of $d$-dimensional quantum systems and introduce qudit hypergraph states. All of the results presented here can be found in Refs.\,\cite{Guehne_2017} and \cite{Xiong_2018}, where qudit hypergraph states where first studied.

In the following, $d$ will always denote an integer larger than one, and $n$ a positive integer. Moreover, we will use $\omega=e^{2 \pi i /d}$ for the $d$-th complex root of unity, $\mathbb{Z}_d$ for the integers modulo $d$, and $\oplus$ for the addition modulo $d$.

\subsection{$d$-dimensional $n$-partite systems}

We consider an $n$-partite quantum system described by the Hilbert space $\mathcal H = \otimes_{i=1}^n \mathcal H_i$, with all the local spaces $\mathcal H_i$ having finite dimension $d$.

Let us focus at first on the local Hilbert spaces. We denote the one-qudit computational basis as $\{|0\rangle,|1\rangle,\dots,|d-1\rangle \}$. We define the local unitary operators 
\begin{equation}
    X = \sum_{q=0}^{d-1} | q \oplus1 \rangle \langle q |\,,\quad
    Z = \sum_{q=0}^{d-1} \omega^q | q \rangle \langle q |\,,
\end{equation}
that generate the $d$-dimensional Pauli group. These have the properties $X^d = Z^d= \mathbb{I}$ and $X^m Z^n = \omega^{-mn} Z^n X^m$. One may notice that for arbitrary dimension they are not self-adjoint in general, whereas for $d=2$ they are and they correspond to the Pauli matrices.

Moving to the full Hilbert space, we define controlled operators among multiple qudits as follows.
Given a local operator $M$ and a set of $r$ qudits with indices $\mathcal I = \{i_1, i_2, \dots, i_r \}$,
the $r$-qudit controlled operator for $M$ is obtained recursively as
\begin{equation}
    M_{\mathcal I} = \sum_{q=0}^{d-1} | q_{i_1} \rangle \langle q_{i_1} | \otimes M_{\mathcal I \setminus \{i_1\} } \,. 
\end{equation}
This procedure permits to build the $r$-qudit controlled phase gate $Z_{\mathcal I}$ from the local unitary $Z$, obtaining 
\begin{equation}
    Z_{\mathcal I} = \sum_{q_{i_1}=0}^{d-1} \dots \sum_{q_{i_r}=0}^{d-1}
    \omega^{q_{i_1}\cdot ... \cdot q_{i_r}} \,
    |q_{i_1}  \dots q_{i_r} \rangle \langle q_{i_1}  \dots q_{i_r} | \,.
\end{equation}
This has again the property $Z_{\mathcal I}^d = \mathbb{I}$, and different controlled phase operators commute.
Moreover, it can be shown that the following identities involving the $X$ Pauli operator hold:
\begin{equation} \label{x_and_z}
    X_k^\dagger Z_{\mathcal I} X_k = Z_{\mathcal I \setminus \{ k \} } Z_{\mathcal I} \,, \quad
    Z_{\mathcal I}^\dagger X_k Z_{\mathcal I}  = X_k Z_{\mathcal I \setminus \{ k \} } \,.
\end{equation}

\subsection{Qudit hypergraph states}

A multi-hypergraph is a pair $H=(V,E)$, where $V$ is a set of vertices and $E$ is a multi-set of hyperedges (a hyperedge is any subset of $V$). In the following, we will always consider non-empty hyperedges. Being $E$ a multiset, a hyperedge $e \in E$ may appear more than once: the number of times it appears is its multiplicity, that we denote as $m_e$. Examples of multi-hypergraphs are depicted in Fig.\,\ref{fig:examples}.

\begin{figure}[t]
    \centering
    \includegraphics[width=0.95\textwidth]{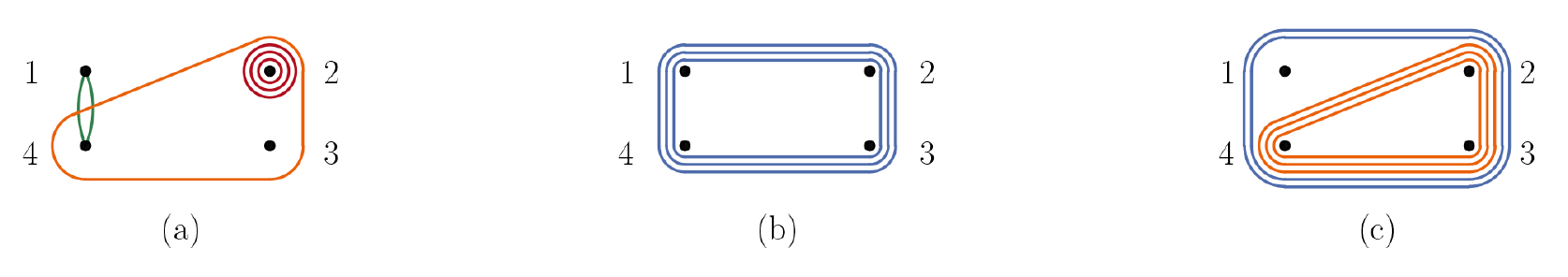}
    \caption{Examples of multi-hypergraphs $H=(V,E)$. Points represent vertices and lines hyperedges. In all cases $V=\{1,2,3,4\}$. In (a), $E=\{ \{2,3,4\}, \{1,4\}^2, \{2\}^3 \}$, with a multiplicity-3 loop on vertex 2. In (b), $E=\{ \{1,2,3,4\}^3 \}$. This is an elementary multi-hypergraph. In (c), $E=\{ \{1,2,3,4\}^2, \{2,3,4\}^3 \}$. The associated quantum states are respectively $Z_{\{2,3,4\}} Z_{\{1,4\}}^2 Z_{\{2\}}^3 | + \rangle^V$, $Z_{\{1,2,3,4\}}^3 | + \rangle^V$, and $Z_{\{1,2,3,4\}}^2 Z_{\{2,3,4\}}^3 | + \rangle^V$.}  
    \label{fig:examples}
\end{figure}

\begin{defi}[Qudit hypergraph state]
Given a dimension $d$, we build the quantum state associated to a multi-hypergraph $H=(V,E)$ as follows
\begin{itemize}
    \item We associate a local state $| + \rangle = d^{-1/2} \sum_{q=0}^{d-1} | q \rangle$ to each vertex in $V$. This gives the global state $| + \rangle^V = \otimes_{i \in V} | + \rangle_i$, corresponding to the empty multi-hypergraph.
    \item For each hyperedge $e \in E$, we apply a controlled phase gate $Z_e$ for $m_e$ times to $| + \rangle^V$.
\end{itemize}
The procedure yields the qudit multi-hypergraph state
\begin{equation}
    | H \rangle = \prod_{e \in E} Z_e^{m_e} | + \rangle^V \,.
\end{equation}
\end{defi}
From now on, we will always refer to qudit multi-hypergraph states simply as hypergraph states. Whenever referring to the special case of qubits, we will explicitly specify that. Moreover, we will use attributes (such as for example elementary or connected) indifferently for a multi-graph and the hypergraph state to it associated. 

Since $Z_e^d = \mathbb I$, for hypergraph states we typically consider $m_e \in \mathbb Z_d$.
Please notice that single-vertex hyperedges (known as loops) are possible, see for example Fig.\,\ref{fig:examples}\,(a). In that case, the controlled phase gate $Z_e$ reduces to the local unitary $Z$.

An important class of hypergraph states is represented by states in the simple form $Z_e^{m_e} | + \rangle^V$, namely those endowed with a single maximum-cardinality hyperedge, Fig.\,\ref{fig:examples}\,(b). In accordance with Ref.\,\cite{Guehne_2017}, we call these elementary hypergraph states. In Sec.\,\ref{sect:ent_formula} we compute their multipartite entanglement.

The notion of hypergraph state can be used to define an orthonormal basis of the global Hilbert space. This is known as the hypergraph state basis.
\begin{defi}[Hypergraph state basis] \label{hg_basis}
Let $| H \rangle$ be an $n$-qudit hypergraph state, the states
\begin{equation} \label{hypergraph_basis}
    |H_{k_1,k_2,\dots, k_n} \rangle := Z_1^{-k_1} Z_2^{-k_2} \dots Z_n^{-k_n}  | H \rangle  \,,
\end{equation}
with $k_i \in \mathbb{Z}_d$ $\forall i$, represent an orthonormal basis for the $n$-partite Hilbert space.
\end{defi}

For the purposes of this work, it is interesting to have a closer look to the effect of Pauli operators on hypergraph states. Let us consider the unitary $X$. The application of an $X^\dagger$ gate on a qudit $k$ of a generic hypergraph state, see Eq.\,\ref{x_and_z}, gives
\begin{equation} \label{x_effect}
    X_k^\dagger \prod_{e \in E} Z_e^{m_e} | + \rangle^V =
    \prod_{e \in E} Z_e^{m_e} Z_{e \setminus \{k\}}^{m_e} | + \rangle^V \,.
\end{equation}
This amounts to creating, for each hyperedge containing $k$, an additional hyperedge with the same multiplicity but not containing $k$.
Let us consider now the unitary $Z$. As already observed, it is not self-adjoint for arbitrary dimension, and a measurement of $Z$ itself is not possible. One may measure instead an observable $M_Z = \sum_q \mu_q | q \rangle \langle q |$, with $\mu_q$ real. It can be shown that a measurement of $M_Z$ on qudit $k$ with outcome $\mu_q$ projects a generic hypergraph state $\prod_e Z_e^{m_e} | + \rangle^V$ onto the state
\begin{equation} \label{z_effect}
    |q_k \rangle \otimes 
    \prod_e Z_{e \setminus \{k\}}^{m_e \cdot q} | + \rangle^{V \setminus \{k\}} \,.
\end{equation}
In this process the hyperedges containing $k$ are replaced by new hyperedges not containing $k$ and with a multiplicity that depends on the outcome of the measurement.

Before concluding this section, it is worth mentioning that a characterization of hypergraph states in terms of stabilizer operators is also possible \cite{Guehne_2017}. We do not introduce it here, since the use of stabilizers is not within the scope of this work.

\section{Multipartite entanglement} \label{sect:intro_ent}

We briefly recall definitions and well-known results concerning both bipartite and multipartite entanglement. For a detailed discussion of these topics, we refer to Ref.\,\cite{horodecki_2009}.

\begin{defi}[Bipartite and multipartite entanglement]
    Let $| \psi_n \rangle \in \mathcal H_n$ be a pure state composed of $n$ subsystems. And let  $AB$ be a generic bipartition of the system into $A=\{1,2,\dots,k\}$ and $B=\{k+1,\dots,n\}$, with $1\le k<n$. We define the bipartite entanglement of $| \psi_n \rangle$ with respect to the bipartition $AB$ as
    \begin{equation} \label{def_bipartite}
        E^{AB} ( | \psi_n \rangle ) := 1 - 
        \max_{|\phi^A \rangle | \phi^B \rangle}
        \left| ( \langle \phi^A | \langle  \phi^B | ) | \psi_n \rangle \right|^2 \equiv 
        1 - \alpha^{AB} ( | \psi_n \rangle ) \,,
    \end{equation}
    where the maximum is over all the pure states in $\mathcal H_n$ which are separable with respect to the bipartition $AB$. We define the multipartite entanglement of $| \psi_n \rangle$ as the minimum bipartite entanglement with respect to all possible bipartitions $AB$, namely as
    \begin{equation} \label{def_multipartite}
        E ( | \psi_n \rangle ) := \min_{AB} E^{AB} ( | \psi_n \rangle ) =
        1 - \max_{AB} \alpha^{AB} ( | \psi_n \rangle ) 
        \equiv 1 - \alpha ( | \psi_n \rangle ) \,.
    \end{equation}
    Here $\alpha^{AB} ( | \psi_n \rangle )$ and $\alpha ( | \psi_n \rangle )$ denote respectively the maximum overlap with the pure states in $\mathcal H_n$ which are separable with respect to $AB$, and the maximum overlap with the full set of pure biseparable states in $\mathcal H_n$.
\end{defi}

Interestingly, $\alpha ( | \psi_n \rangle )$ does not need to be calculated by direct maximization of the overlap, but it can also be obtained as the maximum eigenvalue of the reduced states of $| \psi_n \rangle$ \cite{bourenanne_2004}. In Sec.\,\ref{sect:ent_formula} we use this method to calculate the maximum overlap with pure biseparable states (and then the multipartite entanglement) of elementary hypergraph states.

Moreover, as demanded to any good measure of entanglement, both bipartite and multipartite entanglement are invariant under local unitaries (LU) and non-increasing under local operations and classical communication (LOCC) \cite{Plenio2007AnIT}. This is fundamental for the procedure that we develop in Sec.\,\ref{sect:lower_bound}.

\section{String-sets} \label{sect:pi_sets}

In this section we develop a formalism that will be a convenient framework for the entanglement calculations of Sec.\,\ref{sect:ent_formula}. It is based on sets of strings that we call ``string-sets". 
In order to reduce the extent of technicalities in the discussion, all of the results concerning the cardinalities of string-sets are provided here without a proof. The interested reader shall refer to the appendices for details. In Appendix A the general formula for cardinality is derived, whereas Appendix B deals with properties of the cardinalities. The connection between string-sets and quantum states is shown at the end of the section.

\subsection{Definitions}

We give the fundamental definitions of the formalism. 
As a general remark, please notice that in many cases, for the sake of simplicity in the notation, the dimension $d$ is not explicitly marked on mathematical objects. For example, we use $S(n)$ instead of a more informative but heavier $S_d(n)$, or $\zeta_x(n)$ instead of $\zeta_{x,d}(n)$. This should not lead to ambiguities in any case.
Moreover, we use  $\gcd$ for the greatest common divisor and $x|y$ for ``$x$ is a divisor of $y$".

\begin{defi}[$n$-string]  \label{def_nstring}
We define an $n$-string $s$ as a collection of $n$ non-negative integers smaller than $d$:
\begin{equation}
s := \{q_1, q_2, \dots , q_n \}, 
\quad 
\text{with} 
\quad
q_i \in \mathbb{Z}_d \;\; \forall i \; .
\end{equation}
We call $S(n)$ the full set of $n$-strings. 
\end{defi}
One may notice that the cardinality of $S(n)$ is $|S(n)|=d^n$.
\begin{defi}[String-set] \label{def_pi_set}
We define $\zeta_x(n)$ as the set of $n$-strings such that the product of their elements is congruent $\bmod \, d$ to an integer $x$:
\begin{equation}
\zeta_x (n) := \left\{ s \in S(n) \; : \; 
\prod_{i = 1}^n s[i] \equiv x 
\mod d \right\} \,, 
\end{equation}
where $s[i]$ denotes the $i$-th element of the string $s$.
\end{defi}
It follows naturally, for example, that $\zeta_x (n) = \zeta_{x+d} (n)$ $\forall x$ (the freedom in the choice of indices implied by this equality is particularly useful for the calculations in the Appendix). Moreover, $\cup_{i=0}^{d-1} \zeta_i (n) = S(n)$ and $\sum_{i=0}^{d-1} |\zeta_i (n)| = |S(n)|=d^n$. 

\subsection{Cardinalities}

Deriving the analytical expression for the cardinality of string-sets is a straightforward task for $d$ prime, but it involves a certain amount of technicalities for $d$ non-prime. We present here the general formula and refer to Appendix A for the derivation. 

\begin{thm}[Cardinality of string-sets] \label{thm_card_generic}
Let us consider a dimension $d=\prod_{i \in \mathcal J} p_i^{\ell_i}$, with $\{ p_{i \in \mathcal J} \}$ the prime factors of $d$, and a non-negative integer $x$ such that $gcd(x,d)=\prod_{i \in \mathcal J} p_i^{k_i}$. $\mathcal J$ is the union of two sets: $\mathcal J_1$, containing the indices $i$ such that $k_i < \ell_i$, and $\mathcal J_2$, containing those such that $k_i = \ell_i$. The cardinality of the set $\zeta_x (n)$ is
\begin{gather} \label{main_pi_x_card}
\begin{split}
    | \zeta_x (n) | 
    = &\varphi^{n-1} (d) 
    \left[ \prod_{i \in \mathcal J_1} \binom{n+k_i-1}{k_i} \right]
    \left[ \prod_{i \in \mathcal J_2} 
    \sum_{j=0}^{n-1} 
    \binom{n+\ell_i-2-j}{\ell_i-1}
    \left( 1 - \frac{1}{p_i}  \right)^{-j}
    \right] \,,
\end{split}
\end{gather}
where $\varphi$ is Euler's totient function\footnote{We recall that, for a $d$ as in the theorem,
$ \displaystyle
    \varphi \left( d \right) =
    \varphi \left( \prod_{i \in \mathcal  J} p_i^{\ell_i} \right)=
    \prod_{i \in \mathcal  J} \varphi \left(  p_i^{\ell_i} \right) =
    d \, \prod_{i \in \mathcal  J} \left( 1 - \frac{1}{p_i} \right) \,.
$}.
\end{thm}

Such expression, in a few special cases, takes a much simpler form:
\begin{itemize}
    \item For any $x$, $|\zeta_x(1)|=1$.
    \item For $x$ coprime to $d$,
    \begin{equation} \label{pi_coprime}
    | \zeta_{x} (n) | = \varphi^{n-1}(d) 
    =   d^{n-1} \prod_{i \in \mathcal J}  \left( 1 - \frac{1}{p_i}  \right)^{n-1} \,.
    \end{equation}
    \item For $d$ prime,
    \begin{equation}
    | \zeta_{x} (n) | = \varphi^{n-1}(d) = (d-1)^{n-1} \quad \text{if } x \not\equiv 0 \mod d\,,
    \end{equation}
    and
    \begin{equation}
    | \zeta_{x} (n) | = | \zeta_{0} (n) | = 
    d^n - \sum_{x=1}^{d-1} | \zeta_{x} (n) | = d^n - (d-1)^{n}   \quad \text{if } x \equiv 0 \mod d\,.
    \end{equation}
\end{itemize}

It is also worth noticing that, for an $x$ as in Theor.\,\ref{thm_card_generic}, the cardinality of the set $ \zeta_{x} (n) $ can be factorized as
\begin{equation} \label{pi_factorization}
    | \zeta_{x} (n) | = \left| \zeta_{ \prod_{i \in \mathcal J} p_i^{k_i} } (n) \right| = 
    | \zeta_1 (n) |
    \prod_{i \in \mathcal J} \frac{ \left| \zeta_{ p_i^{k_i} } (n) \right|}{| \zeta_1 (n) |} \,.
\end{equation}

\subsection{Properties of cardinalities}  \label{pi_properties}

It can be shown  that the following identities involving the cardinalities of string-sets hold
(see Appendix B for details)

\begin{itemize}
    \item For $\alpha$ and $0<k<n$ integers,
\begin{equation}  \label{k_identity}
    \sum_{r=0}^{d-1} \omega^{\alpha \, r} \vert \zeta_r (n) \vert  
    = \sum_{x=0}^{d-1} \sum_{y=0}^{d-1} \omega^{\alpha \, xy}
    \vert \zeta_{x} (k) \vert \vert \zeta_{y} (n-k) \vert \, .
\end{equation}
\item For $\alpha$  and $n>1$ integers, 
\begin{equation} \label{real_positive} 
    d \, \vert \zeta_0(n-1) \vert \le 
    \sum_{r=0}^{d-1} \omega^{\alpha \, r} \vert \zeta_r (n) \vert  \le 
    d^n     \,.
\end{equation}
Notice that the sum in $r$ is a real number.
\item For a dimension $d= \prod_{i \in \mathcal J} p_{i}^{\ell_i}$, with $\{p_{i \in \mathcal J}\}$ the prime factors of $d$, and $m \in \mathbb{Z}_d$ such that $\gcd(m,d)=\prod_{i \in \mathcal J} p_{i}^{k_i}$,
\begin{equation} \label{alpha_identity}
   \frac{1}{d^n} 
  \sum_{y=0}^{d-1}  \sum_{t=0}^{d-1} \omega^{-m \, ty}  \vert \zeta_{y} (n-1) \vert =
   \prod_{i \in \mathcal J} 
    \left[ 1- (1- \delta_{k_i, \ell_i} ) \left( 1 - \frac{1}{p_i} \right)^{n}
    \sum_{j=0}^{\ell_i-k_i-1} p_i^{-j} \binom{n+j-1}{j}
    \right] \,,
\end{equation}
where $\delta$ denotes the Kronecker delta.
\end{itemize}

These three results will be used in the following section, in the proof of Theor.\,\ref{elementary_entanglement}.

\subsection{$n$-strings and quantum states}
We observe that, for given $d$, there is a one-to-one correspondence between $n$-strings in $S(n)$ and elements of the $n$-qudit computational basis. It seems therefore natural to denote a state $|q_1 q_2 \dots q_n \rangle$ as $|s\rangle$, using the $n$-string $s = \{q_1, q_2, \dots, q_n\} \in S(n)$. In the following, we will use this notation.

\section{Multipartite entanglement of elementary hypergraph states} \label{sect:ent_formula}

In this section, we calculate multipartite entanglement of elementary hypergraph states. These are hypergraph states endowed with a single maximum-cardinality hyperedge, in the form 
\begin{equation}
    | G_n^{m_e} \rangle = Z_e | + \rangle^V = \sum_{q_1=0}^{d-1} \dots \sum_{q_n=0}^{d-1}
    \omega^{q_1 \cdot ... \cdot q_n} | q_1 \dots q_n \rangle \,.
\end{equation}

Throughout all the section, we make use of the formalism and the results presented in Sec.\,\ref{sect:pi_sets}. We also make use of the following matrix norm.

\begin{defi}[Infinity norm]
Let $M \in \mathbb C^{n \times n}$ be a square matrix. We define its infinity norm as 
\begin{equation}
    \Vert M \Vert_\infty := \max_{i \in \{1,2,\dots,n\} }
    \sum_{j=1}^n | M_{ij} | \,.
\end{equation}
\end{defi}
It can be shown  that, for $M\ge0$, its maximum eigenvalue $\lambda_{\max}$ is upper bounded by the infinity norm as $\lambda_{\max} \le \Vert M \Vert_\infty$ \cite{horn_johnson_2012}. 

Before moving into details, we outline the general procedure behind our proof using a generic $n$-partite state $|\psi_n\rangle$. The reasoning that we follow is analogous to some of the proofs in Ref.\,\cite{Ghio_2017}, where the qubit case was investigated. Here we generalise this derivation to generic dimension.

\begin{itemize}
    \item \textit{I.} We consider an $(n-1)$-to-$1$ bipartition $\bar A = \{1,2,\dots,n-1\}$, $\bar B = \{n\}$. We derive the Schmidt decomposition of $| \psi_n \rangle$ with respect to the bipartition $\bar A \bar B$. The maximal squared Schmidt coefficient $\left(s^{\bar A \bar B}_{\max} (| \psi_n \rangle)\right)^2$ is the maximal eigenvalue of the $(n-1)$-qudit reduced state \cite{nielsen_chuang}.
    \item \textit{II.} We consider an $(n-k)$-to-$k$ bipartition $A = \{1,2,\dots,n-k\}$, $B = \{n-k+1,\dots,n\}$, with $k>1$. We construct the reduced density matrix corresponding to $n-k$ qudits and calculate its infinity norm $\Vert \rho^{(1,2,\dots,n-k)} \Vert_\infty$. For the observation above, this represents an upper bound for the eigenvalues of the reduced state. 
    \item \textit{III.} We observe that $\left( s^{\bar A \bar B}_{\max} (| \psi_n \rangle) \right)^2 \ge \Vert \rho^{(1,2,\dots,n-k)} \Vert_\infty $ $\forall k$. Since the maximum overlap of $|\psi_n \rangle$ with pure biseparable states, $\alpha(|\psi_n \rangle)$, can be obtained as the maximum of all the eigenvalues of the reduced states \cite{bourenanne_2004}, we conclude that $\alpha(|\psi_n \rangle)=\left(s^{\bar A \bar B}_{\max} (| \psi_n \rangle) \right)^2$. We finally obtain multipartite entanglement by definition as $E(|\psi_n \rangle) = 1- \alpha(|\psi_n \rangle)$.
\end{itemize}


\begin{thm} 
[Multipartite entanglement - elementary hypergraph states] \label{elementary_entanglement}
Let us consider a dimension $d=\prod_{i \in \mathcal J} p_i^{\ell_i}$, with $\{ p_{i \in \mathcal J} \}$ the prime factors of $d$. Let $\vert G_n^{m_e} \rangle$ be an $n$-qudit elementary hypergraph state with hyperedge multiplicity $m_e \in \mathbb{Z}_d$ such that $\gcd(m_e, d)=\prod_{i \in \mathcal J} p_{i}^{k_i}$. 
The maximum squared overlap between $\vert G_n^{m_e} \rangle$ and the pure biseparable states is 
    \begin{equation} \label{alpha_general}
    \alpha \left( \vert G_n^{m_e} \rangle \right) = 
    \prod_{i \in \mathcal J} 
    \left[ 1- (1- \delta_{k_i, \ell_i} ) \left( 1 - \frac{1}{p_i} \right)^{n-1}
    \sum_{j=0}^{\ell_i-k_i-1} p_i^{-j} \binom{n+j-2}{j}
    \right] \,,
    \end{equation}
and the multipartite entanglement of $\vert G_n^{m_e} \rangle$ is 
    \begin{equation} \label{entanglement_general}
    E \left( \vert G_n^{m_e} \rangle \right) = 1-  
    \prod_{i \in \mathcal J} 
    \left[ 1- (1- \delta_{k_i, \ell_i} ) \left( 1 - \frac{1}{p_i} \right)^{n-1}
    \sum_{j=0}^{\ell_i-k_i-1} p_i^{-j} \binom{n+j-2}{j}
    \right] \,.
    \end{equation}
\end{thm}

\begin{proof}
We split the proof into three different parts, corresponding to the three points outlined above.
\vspace{\baselineskip}

\noindent \textit{Part I: Maximal squared Schmidt coefficient for ($n-1$)-to-$1$ bipartition.} \\
Let us consider the bipartition $\bar A = \{1,2,...,n-1\}$ and $\bar B = \{n\}$. Denoting the product of the elements of a string $s$ as $P(s)$, the state $\vert G_n^{m_e} \rangle$ can be written as
\begin{gather}
\begin{split}
\vert G_n^{m_e} \rangle
&= \displaystyle \frac{1}{\sqrt{d^n}} 
\sum_{s \in S(n-1) } \sum_{q= 0}^{d-1} 
\omega^{m_e \, P(s) \, q} 
\vert s \rangle \vert q \rangle \\
&= \displaystyle \frac{1}{\sqrt{d^n}} 
\sum_{y =0 }^{d-1}
\sum_{s \in \zeta_y (n-1) } \sum_{q = 0}^{d-1} 
\omega^{m_e \, y \, q} 
\vert s \rangle \vert q \rangle \\
&= \displaystyle \frac{1}{\sqrt{d^n}} 
\sum_{y =0 }^{d-1}
\left( 
\sum_{s \in \zeta_y (n-1) } \vert s \rangle 
\right)
\left( \sum_{q = 0}^{d-1} 
\omega^{m_e \, y q} 
\vert q \rangle \right) \,.
\end{split}
\end{gather}
Introducing the local orthonormal basis $| h_k \rangle := d^{-1/2} \, \sum_{q=0} \omega^{k \, q} | q \rangle $, with $k \in \mathbb Z_d$, we have 
\begin{equation} 
\vert G_n^{m_e} \rangle = \displaystyle \frac{1}{\sqrt{d^{n-1}}}\sum_{y=0}^{d-1}
\left( \sum_{s \in \zeta_y (n-1) } \vert s \rangle \right)  \vert h_{m_e y  \bmod  d} \rangle  \,.
\end{equation}  
For $d$ prime, by varying the index $y$ in the sum, $\vert h_{m_e y \, \bmod  d} \rangle$ covers all of the elements of the local basis. But this is not the case for $d$ non-prime. Whenever two indices $y,y'$ are such that $m_e y \equiv m_e y' \mod  d$, then $\vert h_{m_e y \, \bmod  d} \rangle = \vert h_{m_e y' \, \bmod  d} \rangle $.
By exploiting the identity \cite{ConcreteMathematics}
\begin{equation}
    m_e y \equiv m_e y' \mod  d 
    \quad \Leftrightarrow \quad y \equiv y'  \mod  d/\gcd(m_e,d) \,,
\end{equation}
we express $\vert G_n^{m_e} \rangle$ as 
\begin{equation} 
\vert G_n^{m_e} \rangle = 
\frac{1}{\sqrt{d^{n-1}}} 
\sum_{x =0}^{d/\gcd (m_e,d)-1}
\left( 
\sum_{y =0}^{d-1} \delta_{x, \, m_e y \bmod d} \, 
\sum_{s \in \zeta_y (n-1) } \vert s \rangle 
\right) \vert h_x \rangle \, .
\end{equation}
This equation, with the introduction of normalization factors, leads to the Schmidt decomposition 
\begin{equation} 
\vert G_n^{m_e} \rangle =  
\sum_{x =0}^{d/\gcd (m_e,d)-1} 
\frac{ \sqrt{\mathcal N_x} }{\sqrt{d^{n-1}}} 
\cdot 
\frac{\sum_{y = 0}^{d-1} \delta_{x, \, m_e y \bmod d} \, \sum_{s \in \zeta_y (n-1) } \vert s \rangle }{ \sqrt{\mathcal N_x}} \vert h_x \rangle \, ,
\end{equation}
where 
\begin{equation}
    \mathcal N_x = \sum_{y = 0}^{d-1} \delta_{x, \, m_e y \bmod d} \, \vert \zeta_{y} (n-1) \vert \,.
\end{equation}
It is interesting to notice that the Schmidt rank is $d/\gcd (m_e,d)$, in accordance with the results in Ref.\,\cite{Guehne_2017}. In order to identify the largest squared Schmidt coefficient $\sqrt{\mathcal N_x} / \sqrt{d^{n-1}}$, we should maximize $\mathcal N_x$. Since, for a sum of powers of complex roots of unity \cite{Ledermann_1967} the following identity holds
\begin{equation} \label{delta}
    \delta_{x, \, m_e y \bmod d} = 
    \frac{1}{d} \sum_{t=0}^{d-1} \omega^{(m_e y-x)t} \,,
\end{equation}
$\mathcal N_x$ can be reshaped as
\begin{gather}
    \begin{split}
        \mathcal N_x
        &= \frac{1}{d} \sum_{t=0}^{d-1} \omega^{-xt} 
        \sum_{y =0}^{d-1} \omega^{m_e \, yt} \vert \zeta_{y} (n-1) \vert \,.
    \end{split}
\end{gather}
This expression makes it easier to see that, being the sum over $y$ real positive for any $t$ (Eq.\,\ref{real_positive}), $\mathcal N_x$ is maximal for $x=0$. The largest squared Schmidt coefficient is then 
\begin{equation} \label{largest_schmidt}
\left( s_{max}^{\bar A \bar B}(\vert G_n^{m_e} \rangle) \right)^2 =
\frac{1}{d^n} 
  \sum_{y=0}^{d-1}  \sum_{t=0}^{d-1} \omega^{m_e \, ty}  \vert \zeta_{y} (n-1) \vert  \, .
\end{equation}

\vspace{\baselineskip}

\noindent \textit{Part II: Infinity norm for $(n-k)$-qudit reduced density matrix.}\\
Let us consider a bipartition $A = \{1,2...,n-k\}$ and $B=\{n-k+1,...,n \}$, with $k > 1$.
The state $\vert G_n^{m_e} \rangle$ can be written as
\begin{gather}
\begin{split}
\vert G_n^{m_e} \rangle 
&= 
\frac{1}{\sqrt{d^n}} 
\sum_{s \in S (n-k)} \sum_{s' \in S (k)}
\omega^{m_e \, P(s) P(s')} 
\vert s \rangle \vert s' \rangle \\
&= 
\frac{1}{\sqrt{d^n}} 
\sum_{x=0}^{d-1} \sum_{t=0}^{d-1}
\sum_{s \in \zeta_x (n-k)} \sum_{s' \in \zeta_t (k)}
\omega^{m_e \, x t} 
\vert s \rangle \vert s' \rangle \\
&= 
\frac{1}{\sqrt{d^n}} 
\sum_{x=0}^{d-1} \sum_{t=0}^{d-1}
\omega^{m_e \, xt} 
\sum_{s \in \zeta_x (n-k)} 
\sum_{s' \in \zeta_t (k)} 
\vert s \rangle \vert s' \rangle \,,
\end{split}
\end{gather}
and its density matrix as
\begin{equation}
\rho^{(1, \dots, n)} = \displaystyle \frac{1}{d^n} \sum_{x,y,t,u=0}^{d-1} \omega^{m_e \, ( x  t - y  u )} 
\sum_{s \in \zeta_x (n-k)} 
\sum_{s' \in \zeta_t (k)} 
\sum_{s'' \in \zeta_y (n-k)} 
\sum_{s''' \in \zeta_u (k)} 
\vert s \rangle   
\langle s'' \vert 
\otimes
\vert s' \rangle 
\langle s''' \vert \, . 
\end{equation}
By performing the partial trace over the last $k$ qudits, we obtain the reduced density matrix
\begin{gather}
\begin{split}
\rho^{(1, \dots, n-k)} &= \Tr_{n-k,...,n}[\rho^{(1, \dots, n)}] \\ &= \frac{1}{d^n} \sum_{x,y,t=0}^{d-1} \omega^{m_e \, t (x -  y)} |\zeta_t (k)|
\sum_{s \in \zeta_x (n-k)} \sum_{s'' \in \zeta_y (n-k)} 
\vert s \rangle   
\langle s'' \vert 
\, .
\end{split}
\end{gather}
This can be visualized as a block matrix in the form
\begin{gather}
\begin{split}
\rho^{(1, \dots, n-k)} = \frac{1}{d^{n}}
\begin{bmatrix}
    \gamma {(0,0)} \cdot J_{0,0}   &  \gamma {(0,1)} \cdot J_{0,1}  & \dots   &  \gamma {(0,d-1)} \cdot J_{0,d-1} \\
      \gamma {(1,0)} \cdot J_{1,0} &  \gamma {(1,1)} \cdot  J_{1,1} &  \dots   & \gamma {(1,d-1)} \cdot J_{1,d-1}\\
    \vdots 				     &   \vdots  & \ddots   	      & \vdots \\
     \gamma {(d-1,0)} \cdot J_{d-1,0} &  \gamma {(d-1,1)} \cdot J_{d-1,1}  & \dots 		  & \gamma {(d-1,d-1)} \cdot  J_{d-1,d-1}\\
\end{bmatrix} \,,
\end{split}
\end{gather}
where $J_{x,y}$ denotes a matrix of ones with $\vert \zeta_x (n-k) \vert  $ rows and $ \vert \zeta_y (n-k) \vert$ columns, and $\gamma {(x,y)}=\sum_{t=0}^{d-1} \omega^{m_e \, t(x - y)} \vert \zeta_{t} (k) \vert $. All of the elements of $\rho^{(1, \dots, n-k)}$ are real positive (Eq.\,\ref{real_positive}), and its infinity norm is 
\begin{gather}
\begin{split}
\Vert \rho^{(1, \dots, n-k)}\Vert_{\infty} &= 
\frac{1}{d^n} 
\max_{x \in \{0,\dots, d-1\}}
\left\{  \sum_{y=0}^{d-1} \vert \gamma {(x,y)} \vert   
\vert \zeta_{y} (n-k) \vert \right\} \\
&= \frac{1}{d^n} 
\max_{x \in \{0,\dots, d-1\}}
\left\{  \sum_{y=0}^{d-1}  \sum_{t=0}^{d-1} \omega^{m_e \, t(x-y)} \vert \zeta_t (k) \vert 
\vert \zeta_{y} (n-k) \vert \right\} \,.
\end{split}
\end{gather}
In a similar way as done above for $\mathcal N_x$, it can be shown that the expression inside the curly brackets is maximal for $x=0$. Then, using Eq.\,\ref{k_identity}, and recalling that $|\zeta_t(1)|=1$ $\forall t$, we finally obtain
\begin{gather} \label{infinity_norm}
\begin{split}
\Vert \rho^{(1, \dots, n-k)}\Vert_{\infty} &=  \frac{1}{d^n} 
  \sum_{y=0}^{d-1}  \sum_{t=0}^{d-1} \omega^{m_e \, ty} \vert \zeta_t (k) \vert 
\vert \zeta_{y} (n-k) \vert  \\
&=  \frac{1}{d^n} 
  \sum_{y=0}^{d-1}  \sum_{t=0}^{d-1} \omega^{m_e \, ty} \vert \zeta_t (1) \vert 
\vert \zeta_{y} (n-1) \vert  \\
&= \frac{1}{d^n} 
  \sum_{y=0}^{d-1}  \sum_{t=0}^{d-1} \omega^{m_e \, ty}  \vert \zeta_{y} (n-1) \vert  \,.
\end{split}
\end{gather}

\vspace{\baselineskip}

\noindent \textit{Part III: Conclusions.} \\
The maximal squared Schmidt coefficient $\left( s_{max}^{\bar A \bar B}(\vert G_n^{m_e} \rangle) \right)^2$ in Eq.\,(\ref{largest_schmidt}), and the infinity norm $\Vert \rho^{(1, \dots, n-k)}\Vert_{\infty}$ in Eq. (\ref{infinity_norm}), are identical. This implies that the inequality 
\begin{equation}
    \left( s_{max}^{\bar A \bar B}(\vert G_n^{m_e} \rangle) \right)^2 \ge 
    \Vert \rho^{(1, \dots, n-k)}\Vert_{\infty}
\end{equation}
is satisfied for every $k$. Therefore, the maximum squared overlap of the state $\vert G_n^{m_e} \rangle$ with pure biseparable states is
\begin{equation}
   \alpha (\vert G_n^{m_e} \rangle) =  
   \left( s_{max}^{\bar A \bar B}(\vert G_n^{m_e} \rangle) \right)^2 = 
   \frac{1}{d^n} 
  \sum_{y=0}^{d-1}  \sum_{t=0}^{d-1} \omega^{m_e \, ty}  \vert \zeta_{y} (n-1) \vert  \,.
\end{equation}
Using the identity in Eq.\,(\ref{alpha_identity}), and recalling that by definition $E (\vert \psi_n \rangle) =  1 - \alpha (\vert \psi_n \rangle)$, we obtain the equations in the assertion of the theorem.

\end{proof}

Entanglement values for small elementary hypergraph states up to $d=10$ are shown in Fig.\,\ref{fig:ent}. These are calculated using the formula in Eq. (\ref{entanglement_general}). In the figure one can observe the different entanglement values determined by multiplicities.

Moreover, the results obtained in Theor.\,\ref{elementary_entanglement} permit a few interesting observations:
\begin{itemize}
    \item Elementary hypergraph states such that $gcd(m_e,d) = gcd(m_e',d)$ have the same entanglement. This is a consequence of them being equivalent under LU, as shown in  \cite{steinhoff_2019}.
    \item For $d$ prime, entanglement of elementary hypergraph states does not depend on the multiplicity of the hyperedge. The formulas above reduce to the simpler form 
    \begin{equation}
       \alpha(\vert G_n^{m_e} \rangle) = 
       1 - \frac{(d-1)^{n-1}}{d^{n-1}} \,, 
       \quad
       E(\vert G_n^{m_e} \rangle) = 
       \frac{(d-1)^{n-1}}{d^{n-1}} \,.
    \end{equation}
    For $d=2$ these are consistent with the results in Ref.\,\cite{Ghio_2017}.
    \item For given dimension and hyperedge multiplicity, entanglement of elementary hypergraph states is strictly monotonically decreasing in the number of qudits. This can be seen in Fig.\,\ref{fig:ent} and is proved analytically in Appendix C.
    \item For given dimension, among the elementary hypergraph states with the same number of qudits, the state having multiplicity equal to $d / \operatorname{lpf}(d)$, where $\operatorname{lpf}$ stands for least prime factor, is minimally entangled. This can be seen either by directly considering Eq.\,(\ref{entanglement_general}) (the minimum is attained when for all indices but one $k_i = l_i$, and the remaining index is the one corresponding to the smallest $p_i$ in the factorization), or as a consequence of Prop.\,\ref{ent_multi} in the Appendix with the choice of $m^*=1$. The states having multiplicity coprime to the dimension are maximally entangled (in that case $k_i=0 \; \forall i$ in Eq.\,\ref{entanglement_general}). Observing that 1 is coprime to any $n$, this leads to the inequality
    \begin{equation}
    E(|G_n^{d / \text{lpf} (d)} \rangle ) \le
    E(|G_n^{m_e} \rangle ) \le
    E(|G_n^{1} \rangle ) \,.
    \end{equation}

\end{itemize}

\begin{figure}[p]
    \centering
    \includegraphics[width=1.\textwidth]{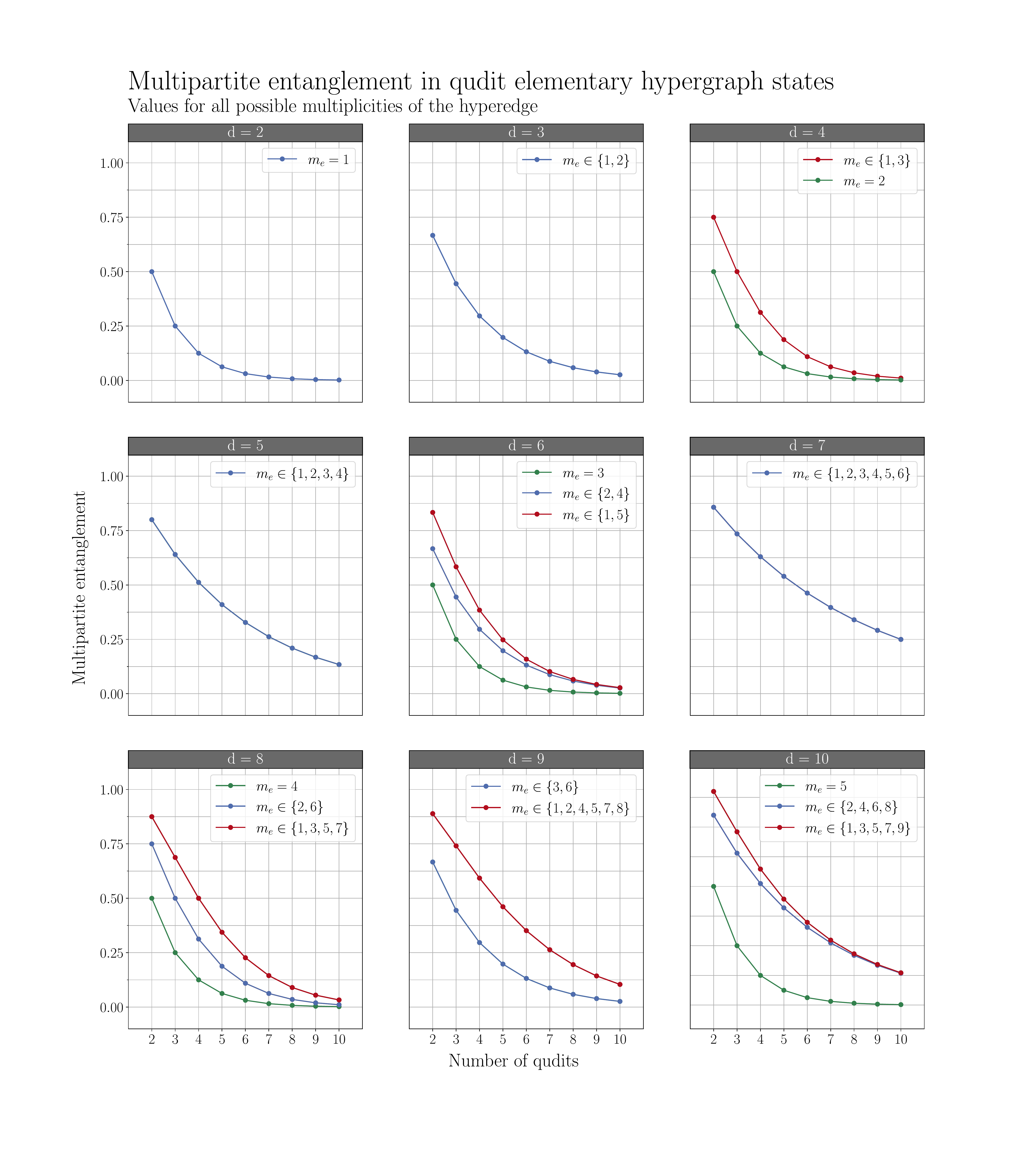}
    \caption{Multipartite entanglement of elementary hypergraph states: values for states up to 10 qudits and up to 10 dimensions (all multiplicities are shown). For $d$ prime, entanglement does not depend on multiplicity. For $d$ non-prime, the curve for minimally entangled states is in green and that for maximally entangled ones in red.}
    \label{fig:ent}
\end{figure}

\section{Lower bound to multipartite entanglement of a generic hypergraph state} \label{sect:lower_bound}

In this section we prove that there exists a lower bound for multipartite entanglement of connected hypergraph states. This is a generalization of Theor.\,III.4 of Ref.\,\cite{Ghio_2017}, where this result was proved for the qubit case. The global structure of our proof is analogous to the one in the reference. There are nevertheless some relevant differences, determined by the generalization to arbitrary dimension. We discuss these more in detail in the following section.

\subsection{Preliminary observations}

As already pointed out, the Pauli operators $X$ and $Z$ are self-adjoint for $d=2$, but not for arbitrary $d$. For this reason, by measurements in the $Z$ basis, in the following we will always mean measurements of a non-degenerate observable $M_Z=\sum_q \mu_q | q \rangle \langle q |$, with $\mu_q$ real. Moreover, we will always apply $X^\dagger$ gates (and not $X$) to hypergraph states.

The removal of hyperedges by applying unitary gates, a procedure systematically used in the reference, requires special care for arbitrary dimension. For $d=2$, it is always possible to remove a largest-minus-one cardinality hyperedge contained into a largest cardinality one by means of an $X^{(\dagger)}$ gate (in that case $X^{\dagger}= X$). For arbitrary $d$, such operation is only possible if certain relations among the dimension and the multiplicities of the considered hyperedges  hold.
Let us consider for example the hypergraph state 
\begin{equation*}
    | H_4 \rangle = Z_{\{1,2,3,4\}}^2 Z_{\{2,3,4\}}^3 |+\rangle^4 \,,
\end{equation*}
for arbitrary $d>3$. The state is depicted in Fig.\,\ref{fig:examples}\,(c). The application of an $X^\dagger$ gate on the first qudit for $q$ times, see Eq.\,\ref{x_effect}, leads to
\begin{equation*}
    (X_1^\dagger)^q | H_4 \rangle = Z_{\{1,2,3,4\}}^2 Z_{\{2,3,4\}}^{3+2q} |+\rangle^4 \,.
\end{equation*}
If, for example, $d=5$, the hyperedge $\{2,3,4\}$ can be removed by applying the gate once ($q=1$). But, if $d=6$, there is no $q$ such that $3+2q \equiv 0 \mod 6$, and the hyperedge cannot be removed.

The following proposition analyzes more in general the case of hypergraph states endowed with a single largest cardinality hyperedge containing at least one largest-minus-one cardinality hyperedge. It inspects the possibility of removing hyperedges by means of $X^\dagger$ gates and the effect of measurements in the $Z$ basis.

\begin{prop} \label{prop_non_separable}
Let $| H_n \rangle$ be an $n$-qudit hypergraph state in the form 
\begin{equation} \label{remove_condition}
    | H_n \rangle = Z_{e}^{m_e} Z_{e \setminus \{k\}}^{m_{e \setminus \{k\}} } | H_n' \rangle \,,
\end{equation}
where $e$ is a hyperedge of cardinality $|e|=n$, $k$ is a qudit, and $| H_n' \rangle$ is either a hypergraph state featuring only hyperedges of cardinality smaller than $n$ and different from $e \setminus \{k\}$ or the empty one. Let also $AB$ be a bipartition crossing both $e$ and $e \setminus \{k\}$. For $d$ prime, it is always possible to remove the hyperedge $e\setminus \{k\}$ by means of an $X^{\dagger}_k$ gate. For $d$ not prime, whenever this operation is not possible, a measurement in the $Z$ basis on the qudit $k$ never results in a separable state with respect to the bipartition $AB$.
\end{prop}
\begin{proof}
The application of the Pauli gate $X^\dagger$ on qudit $k$ for $q$ times gives
\begin{equation}
    (X_k^\dagger)^q | H_n \rangle = Z_{e}^{m_e} 
    Z_{e \setminus \{k\}}^{m_{e \setminus \{k\}} + q m_e } (X^\dagger_k)^q | H_n' \rangle \,,
\end{equation}
where $(X^\dagger_k)^q  | H_n' \rangle$ is again a hypergraph state featuring only hyperedges of cardinality smaller than $n$ and different from $e \setminus \{k\}$ or the empty one.
In order to remove the hyperedge $e\setminus \{k\}$, there has to exist a $q$ satisfying the condition
\begin{equation*}
    m_{e \setminus \{k\}} + q m_e \equiv 0 \mod d\,.
\end{equation*}
Such a $q$ does not always exist \cite{long1965}. It does, for example, if $d$ is prime\footnote{In that case, as a consequence of Lemma \ref{lem_num_classes}, the $d$ possible products obtained multiplying $m_e$ by a $q \in \mathbb Z_d$, belong to $d$ different congruence classes $\bmod \, d$.}. 

Let us consider the case where there is no $q$ satisfying the previous condition. A measurement in the $Z$ basis on qudit $k$, see Eq.\,\ref{z_effect},  has one of the possible outcomes
\begin{equation} \label{after_z}
    | q \rangle Z_{e \setminus \{k\}}^{m_{e \setminus \{k\}}  + q m_e} | H_{n-1}  \rangle \,,
\end{equation}
where $| H_{n-1} \rangle$ is either a hypergraph state featuring only hyperedges of cardinality smaller than $n-1$ or the empty one. Eq.\,\ref{after_z} can result in a separable state with respect to the bipartition $AB$ only if there exists a $q$ such that 
\begin{equation*}
    m_{e \setminus \{k\}} + q m_e \equiv 0  \mod d \,.
\end{equation*}
For our assumptions on $q$, this is never possible. 
\end{proof}

The fact that, for the qubit case, it is always possible to remove a largest-minus-one cardinality hyperedge contained into a largest-cardinality one by means of an $X^{(\dagger)}$ gate, emerges as a consequence of $d=2$ being a prime number.

\subsection{Lower-bound theorem} 

In the light of our previous observations, we develop an iterative a procedure which permits to transform any connected hypergraph state into a state constituted by an elementary hypergraph state and possible additional factorized qudits. A use-case of the procedure is shown in Fig.\,\ref{fig:red}.

\begin{prop} \label{prop_reduction}
Let $| H_{n,k_{max}} \rangle$ be an $n$-qudit connected hypergraph state of largest hyperedge cardinality $k_{max}$. It is always possible to transform $| H_{n,k_{max}} \rangle$ into a state in the form 
\begin{equation}
    | G_{\kappa}^{\mu} \rangle | q_1 \rangle | q_2 \rangle \dots | q_{n - \kappa } \rangle \,,
\end{equation}
by only means of single-qudit measurements and operations that are local with respect to a chosen bipartition $AB$. $| G_\kappa^{\mu} \rangle$ is an elementary hypergraph state with $2 \le \kappa \le k_{max}$ qudits and hyperedge multiplicity $\mu \in \mathbb Z_d$, and it is still crossed by the bipartition $AB$.
\end{prop}
\begin{proof}
We provide a step-by-step iterative procedure which permits to realize the desired transformation:
\begin{itemize}

    \item Among the hyperedges crossing the bipartition $AB$, we select one (or the one) with the largest cardinality. We denote its cardinality as $\kappa$: by construction $\kappa \le k_{max}$.
    
    \item If $\kappa = k_{max}=n$, we move directly to the next point, otherwise we perform measurements in the $Z$ basis on the qudits outside of the selected hyperedge. These measurements do not delete the selected hyperedge, but can modify the structure of hyperedges inside of it. The measurements shall be repeated until the state is in the form
    \begin{equation}
        | H_\kappa \rangle | q_1 \rangle | q_2 \rangle \dots | q_{n - \kappa } \rangle\,,
    \end{equation}
    where $| H_\kappa \rangle$ is a hypergraph state endowed with a single hyperedge of largest cardinality $\kappa$ and possibly other lower cardinality internal hyperedges, and $\{ | q_i \rangle \}$ are single-qudit states whose state depends on the outcomes of the measurements. If $| H_\kappa \rangle$ is an elementary hypergraph state, then we stop here, otherwise we move to the next point.
    
    \item Three scenarios are now possible:
    
    \begin{itemize}
        \item There are no hyperedges of cardinality $\kappa-1$ crossing the bipartition. In this case, we move directly to the next point.
        \item There are one or more hyperedges of cardinality $\kappa-1$ crossing the bipartition, and they can all be removed by means of $X^\dagger$ gates. In this case, we remove them and then move to the next point.
        \item There are one or more hyperedges of cardinality $\kappa-1$ crossing the bipartition which cannot be removed by means of $X^\dagger$ gates. In this case, we measure in the $Z$ basis on a qudit outside of a hyperedge that cannot be removed, but within the highest cardinality hyperedge. For Prop.\,\ref{prop_non_separable}, the outcome of the measurement will never result in a state which is separable with respect to the chosen bipartition. We restart then the procedure from the beginning, using the hypergraph state resulting from the measurement as an input state. Please notice that, in the new round, $\kappa$ will be redefined to a value lower by one.
    \end{itemize}
    
    \item At this point, there are no more hyperedges of cardinality $\kappa-1$, but there can be hyperedges of smaller cardinality. If there are any not crossing the bipartition, we remove them by means of controlled phase gates. Despite not being single-qudit, these operations are local with respect to the bipartition. This guarantees that entanglement is non-increasing under them. For $d>2$, it can be necessary to apply a same controlled phase gate multiple times, according to the multiplicities of the hyperedges to be removed.
    
    \item The state is currently in the form 
    \begin{equation}
        | {H'}_\kappa \rangle | q_1 \rangle | q_2 \rangle \dots | q_{n - \kappa } \rangle\,,
    \end{equation}
    where $| H_{\kappa}' \rangle$ is still a hypergraph state endowed with a single hyperedge of largest cardinality $\kappa$ and possibly other lower cardinality internal hyperedges.
    If $| {H'}_\kappa \rangle$ is an elementary hypergraph state, we stop here. If this is not the case, then there are one or more internal hyperedges of cardinality smaller than $\kappa -1$ crossing the bipartition. Among these, we select one (or the one) of largest cardinality. We measure in the $Z$ basis on a qudit outside of it but within the hyperedge of cardinality $\kappa$. By construction, the state resulting from the measurement possesses at least one hyperedge crossing the bipartition. This state might be a non-connected one. In such case, among the hyperedges crossing the bipartition, we select one of largest cardinality, and we remove all the hyperedges not connected to it by means of (possibly repeated) measurements in the $Z$ basis. 
    We restart then the procedure from the beginning. Please notice that, in the new round, $\kappa$ will be redefined to a lower value.
    
\end{itemize}

The procedure above yields a state in the form 
\begin{equation}
    | {G}_\kappa^\mu \rangle | q_1 \rangle | q_2 \rangle \dots | q_{n - \kappa } \rangle\,,
\end{equation}
where $| {G}_\kappa^\mu \rangle$ is an elementary hypergraph state crossed by the bipartition $AB$, with $2 \le \kappa \le k_{max}$ and $\mu \in \mathbb{Z}_d$.
    
\end{proof}

\begin{figure}[p]
    \centering
    \includegraphics[width=0.9\textwidth]{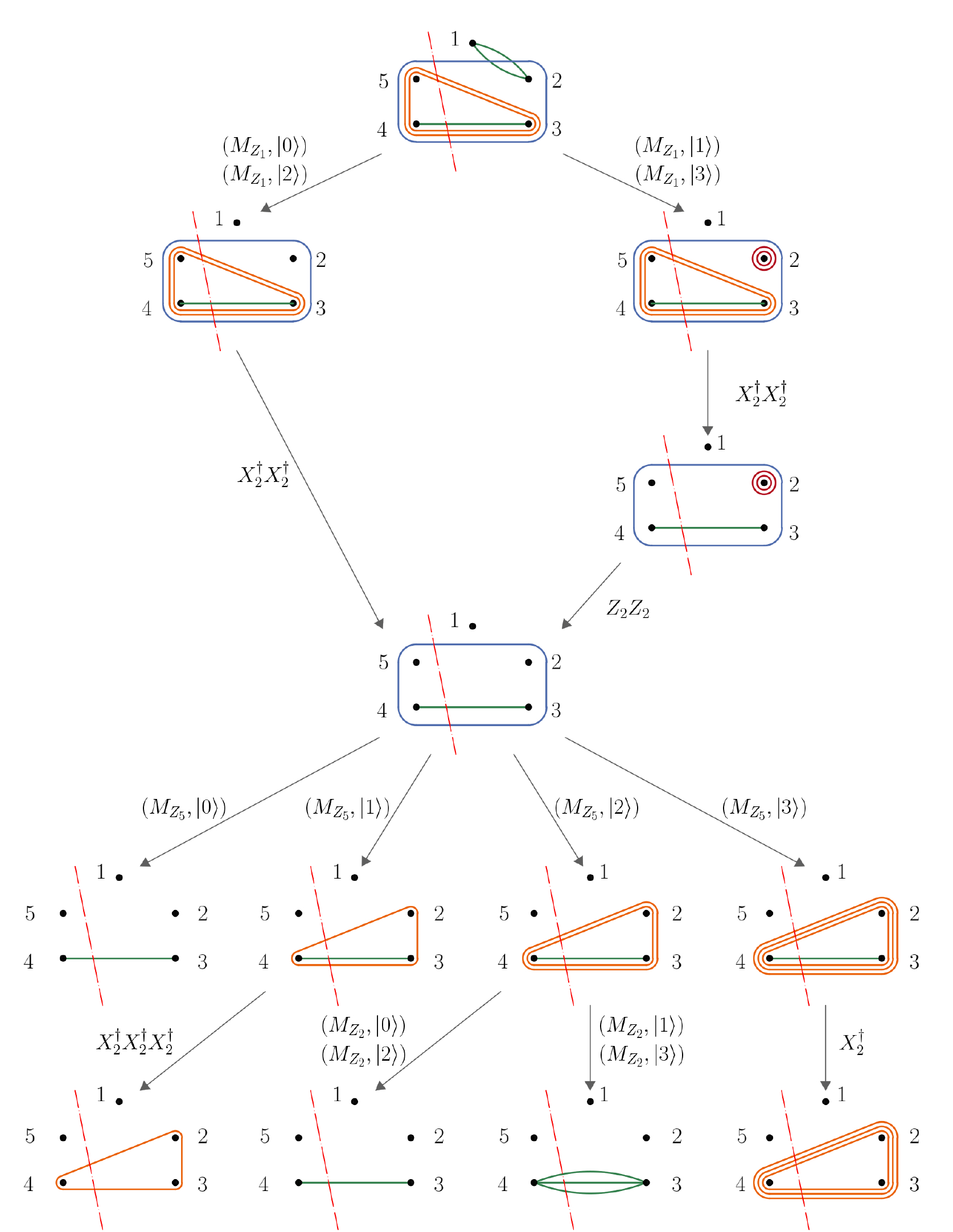}
    \caption{Example of transformation of a generic hypergraph state into an elementary one, following the procedure in Prop.\,\ref{prop_reduction}. Here the dimension is $d=4$ and the input state $Z_{\{2,3,4,5\}}^{} Z_{\{3,4,5\}}^2 Z_{\{1,2\}}^2 Z_{\{3,4\}^{ } } |+\rangle^5$, with the choice of the bipartition $A=\{1,2,3\}$, $B=\{4,5\}$. In the scheme, the notation $(M_{Z_i}, | j \rangle)$ stands for a measurement in the $Z$ basis on qudit $i$ that projects the qudit onto the state $| j \rangle$, whereas $X^\dagger_k$ $(Z_k)$ stands for the application of an $X^\dagger$ $(Z)$ gate to qudit $k$.}
    \label{fig:red}
\end{figure}

Using the previous procedure, we prove the existence of a lower bound for multipartite entanglement of connected hypergraph states.

\begin{thm}[Multipartite entanglement - lower bound] \label{lower_general}
Let $| H_{n,k_{max}} \rangle$ be an $n$-qudit connected hypergraph state of maximum hyperedge cardinality equal to $k_{max}$. Then its overlap with pure biseparable states is upper bounded as
\begin{equation} \label{alphakmax}
    \alpha(| H_{n,k_{max}} \rangle) \le  \alpha(\vert G_{k_{max}}^{d / \text{\normalfont lpf} (d)} \rangle )  \,,  
\end{equation}
and its multipartite entanglement is lower bounded as
\begin{equation} \label{ekmax}
    E(| H_{n,k_{max}} \rangle) \ge  E(\vert G_{k_{max}}^{d / \text{\normalfont lpf} (d)} \rangle )  \,,  
\end{equation}
where $\vert G_{k_{max}}^{d / \text{\normalfont lpf} (d)} \rangle$ is an elementary hypergraph state of cardinality $k_{max}$ and multiplicity $d / \text{\normalfont lpf} (d)$, with $\text{\normalfont lpf} (d)$ the least prime factor of $d$.
\begin{proof}

    Let us consider a bipartition $AB$. Using the procedure detailed in Prop.\,\ref{prop_reduction}, we reduce the state $| H_{n,k_{max}} \rangle$ to an elementary hypergraph state $| G_{\kappa}^{\mu} \rangle$ which is still crossed by the bipartition $AB$ and such that $2 \le \kappa \le k_{max}$ and $\mu \in \mathbb{Z}_d$. Such a reduction is realized by only means of single-qudit measurements and operations that are local with respect to the chosen bipartition $AB$, under which entanglement is non-increasing. Therefore, the initial state is equally or more entangled (with respect to the bipartition $AB$) than the outcome state of the procedure:
    \begin{equation}
        E^{AB} (| H_{n,k_{max}} \rangle) \ge E^{AB} (| G_{\kappa}^{\mu} \rangle)\,.
    \end{equation}
    In particular, the initial state is equally or more entangled than the least-entangled possible outcome state:
    \begin{equation}
        E^{AB} (| H_{n,k_{max}} \rangle) \ge \min_{\mu \in \mathbb{Z}_d, \, 2\le \kappa \le k_{max}} E^{AB} (| G_{\kappa}^{\mu} \rangle) \,.
    \end{equation}
    Using the definition of multipartite entanglement, and noticing that $E (| G_{\kappa}^{\mu} \rangle)$ is minimized by $\mu = d / \text{lpf} (d)$ and $\kappa=k_{max}$ (see the observations after Theor.\,\ref{elementary_entanglement}), we have
    \begin{equation}
    E^{AB} (| H_{n,k_{max}} \rangle) \ge 
    \min_{\mu \in \mathbb{Z}_d, \, 2\le \kappa \le k_{max}} E (| G_{\kappa}^{\mu} \rangle) =
    E (| G_{k_{max}}^{d / \text{lpf} (d)} \rangle ) \,.
    \end{equation}
    Since the inequality holds for any choice of the initial bipartition $AB$, we can conclude that 
    \begin{equation}
    E (| H_{n,k_{max}} \rangle) \ge
    E (| G_{k_{max}}^{d / \text{lpf} (d)} \rangle ) \,.
    \end{equation}
    
    The inequality for the maximum overlap with pure biseparable states follows naturally from the previous one by recalling that $E (| \psi_n \rangle) = 1 - \alpha (| \psi_n \rangle)$.
\end{proof}
\end{thm}

One may observe that for $d$ prime Eqs.\,\ref{alphakmax}-\ref{ekmax} reduce to the simpler forms
\begin{equation}
   \alpha(| H_{n,k_{max}} \rangle) \le 
   1 - \frac{(d-1)^{k_{max}-1}}{d^{k_{max}-1}} \,, 
   \quad
   E(| H_{n,k_{max}} \rangle) \ge 
   \frac{(d-1)^{k_{max}-1}}{d^{k_{max}-1}} \,.
\end{equation}

\subsection{A tighter lower bound}
For a special class of hypergraph states, the lower bound given above in Theor.\,\ref{lower_general}, can be tightened. Let us consider states endowed with a hyperedge of largest cardinality $k_{max}$ and such that the multiplicities of their hyperedges are all multiple of a same divisor of $d$, namely states in the form\footnote{In principle, what follows could be extended to states also featuring single-vertex loops with multiplicities not multiple of $m^*$. For the sake of simplicity, we don't consider this case.}
\begin{equation} \label{all_multiples}
    | H_{n,k_{max}} \rangle = \prod_{e \in E} Z_e^{\beta_e \, m^*} | + \rangle^V \,, \quad
    \beta_e \in \mathbb N_0 \, \forall e \,.
\end{equation}
As a consequence of Eqs.\,\ref{x_effect} and \ref{z_effect}, the procedure described in Prop.\,\ref{prop_reduction}, in this case, can only yield an outcome state in the form
\begin{equation}
    | {G}_\kappa^{ \beta \, m^*} \rangle | q_1 \rangle | q_2 \rangle \dots | q_{n - \kappa } \rangle  \,.
\end{equation}
Then, we can introduce a lower bound for such states as 
\begin{equation} \label{tighter}
E(| H_{n,k_{max}} \rangle) \ge 
\min_{2 \le \kappa \le n, \, 1 \le \beta < d/m^*} E (| G_{\kappa}^{\beta \, m^*} \rangle) =
E(| G^{d/\operatorname{lpf}(d/m^*) }_n \rangle)  \,,
\end{equation}
where the equality comes as a consequence of Prop.\,\ref{ent_mono} and Prop.\,\ref{ent_multi} in the Appendix.

The observation above is of interest whenever $\operatorname{lpf}(d/m^*) \neq \operatorname{lpf}(d)$. In that case, Eq.\,\ref{ekmax} and Eq.\,\ref{tighter} provide different lower bounds. Let us consider for example the state $$| H_{6,4} \rangle = Z_{\{3,4,5,6\}}^{8} Z_{\{1,2,3\}}^2 Z_{\{3,4,5\}}^6 Z_{\{1,5\} }^{4} | + \rangle^V, $$ for $d=10$. Eq.\,\ref{ekmax} gives the inequality
$$
E \left( | H_{6,4} \rangle \right) \ge E \left( | G^5_4 \rangle \right) = 
0.125 \,,$$
whereas Eq.\,\ref{tighter} gives 
$$
E \left( | H_{6,4} \rangle \right) \ge E \left( | G^2_4 \rangle \right) = 
0.512 \,,$$
namely a lower bound more than four times larger.

Notice that Eq.\,\ref{all_multiples} also comprises the special case of uniform-multiplicity hypergraph states. In that case,
\begin{equation}
    | H_{n,k_{max}} \rangle = \prod_{e \in E} Z_e^{m} | + \rangle^V  \,,
\end{equation}
and the lower bound can be expressed as
$$
E(| H_{n,k_{max}} \rangle) \ge E(| G^{d/\operatorname{lpf} \left( 
d / \operatorname{gcd} (m,d) \right)}_n \rangle) \,.
$$

\section{Conclusions}
\label{s:Conclusions}

In this article we investigated the entanglement properties of qudit hypergraph states for an arbitrary number of qudits $n$ and arbitrary finite dimension $d$. We computed multipartite 
entanglement of elementary hypergraph states (those endowed with a single maximum-cardinality hyperedge). We showed that, 
if the multiplicities of two equal-cardinality elementary hypergraph states are such that $\gcd (m_e, d) = \gcd (m_e', d)$, they exhibit the same amount of entanglement, as a consequence of the fact that they are equivalent under LU (as proved in \cite{Guehne_2017}). In particular, we showed that for $d$ prime, all equal-cardinality elementary hypergraph states have the same entanglement, with a simple explicit form for the entanglement content. This supports previous observations that primality of the underlying dimension significantly affects the properties of quantum states \cite{Vourdas_2004}. Moreover, we showed that the multipartite entanglement of a generic hypergraph state is lower bounded by that of an equal-dimension elementary hypergraph state featuring the same number of qudits given by the largest-cardinality hyperedge.

\ack
DM is grateful to Fabio Bernasconi and Claudio Sutrini for useful discussions. AB thanks Alberto Riccardi for fruitful discussions. This research was partly supported by the Italian Ministry of University and Research (MUR) through the ``Dipartimenti di Eccellenza Program (2018-2022)'', Department of Physics, University of Pavia, and by the EU H2020 QuantERA ERA-NET Cofund in Quantum Technologies project QuICHE.

\appendix
\setcounter{section}{0}

\def\thesection{Appendix \Alph{section}}  
\section{Cardinality of string-sets}
\renewcommand{\thesection}{\Alph{section}}

In this section, we derive the formula in Theor.\,\ref{thm_card_generic}. The procedure involves the proof of a few preliminary results, including a recursive formula. We obtain the final result by solving the recursion.

Before starting, we recall some of the observations already reported in Sec.\,\ref{sect:pi_sets} of the main text. 
The definition of string-sets implies that $\zeta_x (n) = \zeta_{x+d} (n)$ $\forall x$, thus leaving freedom in the choice of the indices of the sets. 
We typically consider $x \in \{0,\dots,d\}$. In particular, we use indifferently $\zeta_0 (n)$ and $\zeta_d (n)$, according to what is handier in each specific case. Notice for example that $\sum_{x=0}^{d-1} | \zeta_x (n) | = \sum_{x=1}^d | \zeta_x (n) | = d^n$. We also recall that, by construction, $| \zeta_x (1) | = 1$ $\forall x$.

We start by proving a lemma concerning integer products and congruence classes, and then we derive a first property of the cardinalities of string-sets.

\begin{lem} \label{lem_num_classes}
For $r$ a non-negative integer, and $p$ an integer assuming each value in $\mathbb{Z}_d$ once, the $d$ possible products $r \cdot p$ belong to $d/  \gcd(r,d)$ different congruence classes $\bmod \,d$, in the number of $\gcd(r,d)$ to each class. Such classes can be identified by the values $\gcd(r,d) \cdot i$, with $i \in  \mathbb Z_{d/ \gcd(r,d)}$.
\end{lem}
\begin{proof}
In general, for $r,p,q$ non-negative integers and $d$ a positive one, the following identity holds \cite{ConcreteMathematics} 
\begin{equation} \label{gcd_mod}
    r \cdot p \equiv r \cdot q \mod \, d \quad \Leftrightarrow \quad
    p \equiv q  \mod \, d/\gcd(r,d) \,.
\end{equation}
We prove the two statements in the lemma separately.
\begin{itemize}
    \item By construction, given a non-negative integer $r$, the $d$ integers in $\mathbb{Z}_d$ belong to  $d/ \gcd(r,d)$ different congruence classes $\bmod \,d/ \gcd(r,d)$, in the number of $\gcd(r,d)$ to each class. Then, for the identity above, given a $p$ assuming each value in $\mathbb{Z}_d$ once, the $d$ possible products $r \cdot p$ belong to $d/ \gcd(r,d)$ different congruence classes $\bmod \,d$, in the number of $\gcd(r,d)$ to each class.
    \item Since $r \cdot p$ is either zero or a multiple of $\gcd(r,d)$, and $d$ is by definition a multiple of $\gcd(r,d)$, the residue of the integer division of $r \cdot p$ by $d$ can only be either zero or a multiple of $\gcd(r,d)$. Therefore, the congruence classes $\bmod \, d$ to which the products $r \cdot p$ belong, are those identified by $\gcd(r,d) \cdot i$, with $i \in  \mathbb Z_{d/ \gcd(r,d)}$.
\end{itemize}
\end{proof}

\begin{prop} \label{pi_property}
For any $x$,
\begin{equation} 
    | \zeta_x (n) | =  | \zeta_{\gcd (x,d)} (n) | \,.
\end{equation}
\end{prop}
\begin{proof}
For $n=1$, as well as for $x=\gcd (x)$, and for $x \equiv 0 \mod d$, the proof is trivial. Let us consider $n>1$ and $x$ generic, and let $s$ be a string in $S(n-1)$: by construction it belongs to a set $\zeta_r (n-1)$.
The string $s$ can generate strings in $S(n)$ upon appendage of elements $q_n \in \mathbb Z_d$ at its end. Let us evaluate how many strings it generates in $\zeta_x (n)$ and how many in $\zeta_{\gcd(x,d)} (n)$, when all possible values for $q_n$ are used.
For the previous lemma, if $x \equiv 0 \mod \gcd (r, d)$, $s$ generates $\gcd(r,d)$ strings in $\zeta_x (n)$, otherwise zero; 
if $\gcd(x,d) \equiv 0 \mod \gcd (r, d)$, $s$ generates $\gcd(r,d)$ strings in $\zeta_{\gcd(x,d)} (n)$, otherwise zero.
Since $x / \gcd(x,d)$ is coprime to $d$, then $x \equiv 0 \mod \gcd (r, d) \, \Leftrightarrow \, \gcd(x,d) \equiv 0 \mod \gcd (r, d)$, showing that, for any $x$, $s$ generates the same number of strings in both sets (either $\gcd(r,d)$ or zero).
Being the initial choice of $s$ arbitrary, we can conclude that $| \zeta_x (n) | =  | \zeta_{\gcd (x,d)} (n) |$.
\end{proof}

Using the results above, we derive a recursive formula for the cardinality of string-sets. The derivation of the formula is preceded by the proof of a simple lemma involving the use of Euler's totient function.

\begin{lem} \label{phi_divisors}
Let $d$ be an integer and $\mathcal D(d)$ the set of its divisors. For $r$ an integer assuming each value in $\{1,2, ..., d\}$ once, $\gcd (r,d)$ only assumes values in $\mathcal D(d)$. In particular, it assumes each value $m \in \mathcal D(d)$ for $\varphi (d/m)$ times. 
\end{lem}
\begin{proof} By definition, $\gcd (r,d)$ is a divisor of $d$ for any value of $r$. Let $m$ be a divisor of $d$: then there exists an integer $p$ such that $d=m p $. The equality $\gcd (r,d)=\gcd (r,m p)=m$ only holds if $r$ is in the form $r=m q$, with $q$ an integer coprime to $p$. There are $\varphi(p)=\varphi(d/m)$ such $q$'s. This since $\varphi(p)$ counts the positive integers up to $p$ that are coprime to $p$
\cite{long1965}.
\end{proof}

\begin{prop} \label{prop_x_recursive}
For $n>1$, the cardinality of the set $\zeta_x (n)$ can be expressed recursively as
\begin{equation} \label{recursive_x_eq}
    | \zeta_x (n) | = \sum_{r \in \mathcal C (x,d)}  | \zeta_r (n-1) |  \,  r \, 
    \varphi (d / r) \,,
\end{equation}
where $\mathcal C (x,d)$ is the set of common divisors of $x$ and $d$.
\end{prop}
\begin{proof}
As observed above, a string in $\zeta_r (n-1)$ can generate strings in $\zeta_x (n)$ upon appendage of elements $q_n \in \mathbb Z_d$ at its end only if $x \equiv 0 \mod \gcd (r, d)$. Whenever this condition is met, it can generate $\gcd (r, d)$ such strings. As a consequence, the cardinality of the set $\zeta_x (n)$ can be expressed recursively as
\begin{gather} \label{rec_part}
\begin{split} 
    | \zeta_x (n) | &= \sum_{r =1}^{d}  
    \delta_{0, \, x \bmod \gcd(r,d)} \,
    | \zeta_r (n-1) |  \, \gcd(r,d) \,  \\
    &= \sum_{r =1}^{d}  
    \delta_{0, \, x \bmod \gcd(r,d)} \,
    | \zeta_{\gcd(r,d)} (n-1) |  \, \gcd(r,d) \,,
\end{split} 
\end{gather}
where we encoded the condition on the possibility of generating strings in the Kronecker delta, and we used Prop.\,\ref{pi_property}. 

Let $\mathcal D(d)$ be the set of the divisors of $d$. For the previous lemma, the sum in Eq.\,\ref{rec_part} can be reshaped as sum over the set $\mathcal D(d)$, with the use of a multiplicity factor for each divisor: 
\begin{equation} 
    | \zeta_x (n) | = \sum_{r \in \mathcal D (d)}  
    \delta_{0, \, x \bmod r} \,
    | \zeta_r (n-1) |  \, r \, 
    \varphi (d / r)  \,.
\end{equation} 
The Kronecker delta can be removed by imposing that the sum is over the common divisors of $x$ and $d$.
\end{proof}

We prove a last lemma and then solve the recursion in Eq.\,\ref{recursive_x_eq}.

\begin{lem} \label{lem_mu}
For $d=\prod_{i \in \mathcal J} p_i^{\ell_i}$ and $r=\prod_{i \in \mathcal J} p_i^{k_i}$ positive integers such that $r|d$, where $\{ p_{i \in \mathcal J} \}$ are the prime factors of $d$, it holds that
\begin{equation} \label{eq_mu}
    r \, \varphi (d/r)   = \varphi (d) \, \prod_{i \in \mathcal J} \left( 1 - \frac{1}{p_i}  \right)^{- \delta_{k_i, \ell_i}} \,.
\end{equation}
\end{lem}
\begin{proof}
We introduce two subsets of $\mathcal J$: $\mathcal J_1$, containing the indexes such that $0 \le k_i < \ell_i$, and $\mathcal J_2$, containing those such that $k_i = \ell_i$. By construction $\mathcal J_1 \cup \mathcal J_2 = \mathcal J$. 
We have
\begin{gather*}
\begin{split}
    r \, \varphi (d/r)    &= 
    \prod_{i \in \mathcal J_1} p_i^{k_i} 
    \prod_{j \in \mathcal J_2} p_j^{\ell_j}
    \, \varphi \left(  
    \prod_{i' \in \mathcal J_1} p_{i'}^{\ell_{i'} - k_{i'}}  \right) \\
    &= \prod_{i \in \mathcal J_1} p_i^{k_i} 
    \prod_{j \in \mathcal J_2} p_j^{\ell_j} 
    \prod_{i' \in \mathcal J_1} p_{i'}^{\ell_{i'} - k_{i'}} 
    \left( 1 - \frac{1}{p_{i'}} \right) 
    \\
    &= 
    d \, 
    \prod_{i \in \mathcal J_1} \left( 1 - \frac{1}{p_i} \right) \\  
    &= d \, 
    \prod_{i \in \mathcal J} \left( 1 - \frac{1}{p_i} \right)^{1-\delta_{k_i, \ell_i}} \\   
    &= \varphi (d) \,
    \prod_{i \in \mathcal J} \left( 1 - \frac{1}{p_i} \right)^{- \delta_{k_i, \ell_i}} \,.
\end{split}
\end{gather*}
\end{proof}

One may notice that, if $k_i < \ell_i$ $\forall i$, Eq.\,\ref{eq_mu} reduces to the simpler form  $r \, \varphi (d/r) = \varphi (d)$.

\begin{thm}[Cardinality of string-sets] 
Let us consider $d=\prod_{i \in \mathcal J} p_i^{\ell_i}$, with $\{ p_{i \in \mathcal J} \}$ its prime factors, and $x$ such that $gcd(x,d)=\prod_{i \in \mathcal J} p_i^{k_i}$. $\mathcal J$ is the union of two sets: $\mathcal J_1$, containing the $i$'s such that $k_i < \ell_i$, and $\mathcal J_2$, containing those such that $k_i = \ell_i$. The cardinality of the set $\zeta_x (n)$ is
\begin{gather} \label{pi_x_card}
\begin{split}
    | \zeta_x (n) | 
    = &\varphi^{n-1} (d) 
    \left[ \prod_{i \in \mathcal J_1} \binom{n+k_i-1}{k_i} \right]
    \left[ \prod_{i \in \mathcal J_2} 
    \sum_{j=0}^{n-1} 
    \binom{n+\ell_i-2-j}{\ell_i-1}
    \left( 1 - \frac{1}{p_i}  \right)^{-j}
    \right] \,.
\end{split}
\end{gather}
\end{thm}

\begin{proof}

We introduce $\eta (r) = r \, \varphi (d/r)$, and $\mathcal D (r)$ for the set of the divisors of an integer $r$. By iterating on $n$ in the recursive formula in Prop.\,\ref{prop_x_recursive}, we obtain
\begin{gather} \label{iter}
\begin{split}
    | \zeta_x (n) | 
    &= \sum_{r_1 \in \mathcal C (x,d)}  | \zeta_{r_1} (n-1) |  \, \eta ( r_1 ) \\
    &= \sum_{r_1 \in \mathcal C (x,d)} \sum_{r_2 \in \mathcal D (r_1)}  
    | \zeta_{r_2} (n-2) |  \, \eta (r_1) \, \eta (r_2) \\ 
    &= \sum_{r_1 \in \mathcal C (x,d)} \sum_{r_2 \in \mathcal D (r_1)} \dots 
    \sum_{r_{n-1} \in \mathcal D (r_{n-2})} | \zeta_{r_{n-1}} (1) |  \,
    \eta (r_1) \, \eta (r_2) \dots \eta (r_{n-1}) \\
    &= \sum_{r_1 \in \mathcal C (x,d)} \sum_{r_2 \in \mathcal D (r_1)} \dots 
    \sum_{r_{n-1} \in \mathcal D (r_{n-2})}   \,
    \eta (r_1) \, \eta (r_2) \dots \eta (r_{n-1}) \, ,
\end{split}
\end{gather}
where the sequences of sum indices $r_1, r_2, ...,r_{n-1}$ are such that $r_{t'} | r_t$ for $t' > t$.

Let us focus at first on the special case where $\mathcal J_1 \neq \emptyset$ and $\mathcal J_2 = \emptyset$. In this case, for Lemma\,\ref{lem_mu}, $\eta (r_t) = \varphi (d)$ $\forall t$, which simplifies the previous equation to
\begin{equation}
    | \zeta_x (n) | 
    = \varphi^{n-1} (d) \sum_{r_1 \in \mathcal C (x,d)} \sum_{r_2 \in \mathcal D (r_1)} \dots 
    \sum_{r_{n-1} \in \mathcal D (r_{n-2})}   \, 1 \,.
\end{equation}
Since, in this case, $\gcd (x, d) = \prod_{i \in \mathcal J_1} p_i^{k_i}$, and using the coprimality of the $p_i$'s among themselves, the previous equation can be reshaped as
\begin{equation}
    | \zeta_x (n) | 
    = \varphi^{n-1} (d)  
    \prod_{i \in \mathcal J_1 } 
    \left( \sum_{r_1 \in \mathcal D (p_i^{k_i})}
    \sum_{r_2 \in \mathcal D (r_1)} \dots 
    \sum_{r_{n-1} \in \mathcal D (r_{n-2})}   \, 1 \right)
    \,.
\end{equation} 
In the sum inside the round brackets, the possible sequences of sum indices can be expressed as $r_{1} = p_i^{u_{1}}, r_{2} = p_i^{u_{2}}, \dots, r_{n-1} = p_i^{u_{n-1}}$, with $k_i \ge u_1 \ge u_ 2 \ge \dots \ge u_{n-1} \ge 0 $ (this since $r_{t'} | r_t$ for $t' > t$). Each sequence contributes to the sum with 1. Consequently, performing the sum is equivalent to counting the existing sequences.
For the inequality above, there is only one sequence featuring $N_0$ indices equal to 0, $N_1$ equal to 1, ..., $N_{k_i}$ equal to $k_i$ (with the condition $\sum_{j=0}^{k_i} N_j = n-1$).
Therefore, counting the existing sequences amounts to counting the possible ways of distributing $n-1$ identical objects among $k_i+1$ distinguishable boxes, which is $\binom{n+k_i-1}{k_i}$ \cite{Combinatorics}. Eq.\,\ref{iter} results in
\begin{equation}
    | \zeta_x (n) | 
    = \varphi^{n-1} ( d ) \prod_{i \in \mathcal J_1} \binom{n+k_i-1}{k_i} \,.
\end{equation}

Let us move to the general case. According to Lemma\,\ref{lem_mu}, in Eq.\,\ref{iter}, it is not true any more that $\eta (r_t) = \varphi (d)$ $\forall t$. Instead, for an $r_t=\prod_{i \in \mathcal J} p_i^{j_i} $, 
$\eta(r_t) = \varphi (d) \, \prod_{i \in \mathcal J} 
\left( 1 - \frac{1}{p_i}  \right)^{- \delta_{j_i, \ell_i}}=\varphi (d) \, \prod_{i \in \mathcal J} 
\left( 1 - \frac{1}{p_i}  \right)^{- \delta_{p_i^{j_i}, p_i^{\ell_i}}}$.
Using again the coprimality of the $p_i$'s, and distinguishing between the contribution of the indexes in $\mathcal J_1$ and $\mathcal J_2$, Eq.\,\ref{iter} becomes 
\begin{gather} 
\begin{split}
    | \zeta_x (n) | 
    &= \varphi^{n-1} ( d ) \, \prod_{i \in \mathcal J} \left(  \sum_{r_1 \in \mathcal D (p^{k_i})}
    \sum_{r_2 \in \mathcal D (r_1)}
    \dots 
    \sum_{r_{n-1} \in \mathcal D (r_{n-2})}   \, 
    \left( 1 - \frac{1}{p_i}  \right)^{- \sum_{t=1}^{n-1} \delta_{r_t, p_i^{\ell_i}} } \right) \\
        &= \varphi^{n-1} ( d ) \, \left[ \prod_{i \in \mathcal J_1} \binom{n+k_i-1}{k_i} \right] \cdot \\
    &\cdot \left[ \prod_{i \in \mathcal J_2} \left(  \sum_{r_1 \in \mathcal D (p^{\ell_i})}  
    \sum_{r_2 \in \mathcal D (r_1)} \dots 
    \sum_{r_{n-1} \in \mathcal D (r_{n-2})}   \, 
    \left( 1 - \frac{1}{p_i}  \right)^{- \sum_{t=1}^{n-1} \delta_{r_t, p_i^{\ell_i}}  } \right) \right] \,.
\end{split}
\end{gather}

In the sum inside the round brackets, each sequence of sum indices contributes to the sum with $\left( 1 - \frac{1}{p_i}  \right)^{-j}$, where $j$ is the number of indices equal to $p_i^{\ell_i}$. For each $j$ there exist $\binom{n+\ell_i-2-j}{\ell_i-1}$ sequences, namely the number of ways of distributing $n-1-j$ identical objects among $\ell_i$ distinguishable boxes. Being $j \in \{0, \dots, n-1\}$,  then the sum is equivalent to
\begin{gather} 
\begin{split}
    \sum_{j=0}^{n-1} 
    \binom{n+\ell_i-2-j}{\ell_i-1}
    \left( 1 - \frac{1}{p_i}  \right)^{-j} \,.
\end{split}
\end{gather}
Substituting this result back into the previous equation proves the theorem. 
\end{proof}

\def\thesection{Appendix \Alph{section}}  
\section{Properties of the cardinalities of string-sets}
\renewcommand{\thesection}{\Alph{section}}

We prove the properties of the cardinalities of string-sets listed in Subsec.\,\ref{pi_properties} of the main text. In particular, Prop.\,\ref{lem_equality_pi} contains a proof for Eq.\,\ref{k_identity}, Prop.\,\ref{ineq_omega} one for Eq.\,\ref{real_positive}, and  Prop.\,\ref{horrible_identity} one for Eq.\,\ref{alpha_identity}.

\begin{prop}  \label{lem_equality_pi}
For $\alpha$ and $0<k<n$ integers, it holds that 
\begin{gather} \label{lem_equality_eq}
\begin{split}
\sum_{r=0}^{d-1} \omega^{\alpha \, r} \vert \zeta_r (n) \vert  
&= \sum_{x=0}^{d-1} \sum_{y=0}^{d-1} \omega^{\alpha \, xy}
\vert \zeta_{x} (k) \vert \vert \zeta_{y} (n-k) \vert   \, .
\end{split}
\end{gather}
\end{prop}
\begin{proof}
For $\alpha \equiv 0 \mod d$, $\omega^{\alpha r}=1$, and the proof is straightforward. Let us consider the general case.
We introduce the short-hand notation $P(s)= \prod_{i=1}^n s[i]$ for the product of the elements of a string $s$, and we prove the statement by proving that both sides of Eq.\,\ref{lem_equality_eq} are equal to $\sum_{s \in S(n)} \omega^{\alpha \, P(s)}$.
For the l.h.s. we have
\begin{gather*}
\begin{split}
\sum_{s \in S(n)} \omega^{\alpha \, P(s)}  &= 
\sum_{r=0}^{d-1}  
\sum_{s \in \zeta_r(n)}  
\omega^{\alpha \, P(s)} \\
&= 
\sum_{r=0}^{d-1}  
\sum_{s \in \zeta_r(n)} 
\omega^{\alpha \, r } \\
&= \sum_{r=0}^{d-1}  \omega^{\alpha \, r}
\sum_{s \in \zeta_r (n)} 1 \\
&= \sum_{r=0}^{d-1} \omega^{\alpha \, r}
\vert \zeta_{r} (n) \vert \, ,
\end{split}
\end{gather*}
and for the r.h.s., for any arbitrary choice of a positive integer $k<n$,
\begin{gather*}
\begin{split}
\sum_{s \in S(n)} \omega^{\alpha \, P(s)}  &= 
\sum_{s \in S_k} 
\sum_{s' \in S_{n-k}}  
\omega^{\alpha \, P(s) \, P(s')} \\
&= \sum_{x=0}^{d-1} \sum_{y=0}^{d-1} 
\sum_{s \in \zeta_x (k)} 
\sum_{s' \in \zeta_y (n-k)}  
\omega^{\alpha \, P(s) \, P(s')} \\
&= \sum_{x=0}^{d-1} \sum_{y=0}^{d-1} 
\sum_{s \in \zeta_x (k)} 
\sum_{s' \in \zeta_y (n-k)}  
\omega^{\alpha \, xy} \\
&= \sum_{x=0}^{d-1} \sum_{y=0}^{d-1} \omega^{\alpha \, xy}
\sum_{s \in \zeta_x (k)} 
\sum_{s' \in \zeta_y (n-k)}  
1  \\
&= \sum_{x=0}^{d-1} \sum_{y=0}^{d-1} \omega^{\alpha \, xy}
\vert \zeta_{x} (k) \vert \vert \zeta_{y} (n-k) \vert   \, .
\end{split}
\end{gather*}
\end{proof}

\begin{prop}  \label{ineq_omega}
For $\alpha$  and $n>1$ integers, it holds that 
\begin{equation} \label{ineq_complex}
    d\vert \zeta_0(n-1) \vert \le 
    \sum_{r=0}^{d-1} \omega^{\alpha \, r} \vert \zeta_r (n) \vert  \le 
    d^n     \,.
\end{equation}
\end{prop}
\begin{proof}
Using Prop.\,\ref{lem_equality_pi} with the choice of $k=1$, we have
\begin{gather}
\begin{split}
\sum_{r=0}^{d-1} \omega^{\alpha \, r} \vert \zeta_r (n) \vert  
&= \sum_{x=0}^{d-1} \sum_{y=0}^{d-1} \omega^{\alpha \, xy}
\vert \zeta_{y} (n-1) \vert   \\
&= d \vert \zeta_{0} (n-1) \vert + 
\sum_{y=1}^{d-1} \vert \zeta_{y} (n-1) \vert \sum_{x=0}^{d-1} \omega^{\alpha \, xy} \, .
\end{split}
\end{gather}
Let us focus on the sum $\sum_{x=0}^{d-1} \omega^{\alpha xy}$. If $\alpha y \equiv 0 \mod d$, the sum gives $d$, otherwise it gives $0$ (sum of powers of primitive complex roots of unity) \cite{Ledermann_1967}. The lower bound in Eq.\,\ref{ineq_complex} can be attained for $\alpha$ coprime to $d$, the upper bound for $\alpha \equiv 0 \mod d$.
\end{proof}

\begin{lem}  \label{lem_identity}
For $n,\ell,p$ positive integers, with $p>1$, the following identity holds:
\begin{gather*}
\begin{split}
    p^{-\ell}
    \sum_{j=0}^{n-1} 
    \binom{n+\ell-2-j}{\ell-1}
    \left( 1 - \frac{1}{p}  \right)^{-(j+1)}
    =
    \left( 1 - \frac{1}{p}  \right)^{-n} - \sum_{j=0}^{\ell-1}   
    \binom{n+j-1}{j} p^{-j} \,.
\end{split}
\end{gather*}
\end{lem}

\begin{proof}
We make use of the two identities
\begin{equation} \label{expansion_binomial}
    \left( 1 - \frac{1}{p}  \right)^{-n} = 
\sum_{j=0}^\infty \binom{n+j-1}{j} p^{-j} \quad \text{for} \, n>0 \; \text{and} \; |p|>1\,,
\end{equation}
\begin{equation}
\sum_{j=0}^{r} 
\binom{r+s-j}{s}  \binom{j+k}{k} =
\binom{r+s+k+1}{s+k+1} \,.
\end{equation}

The first one can be easily proved by expanding the l.h.s., whereas a proof for the second one can be found in Ref.\,\cite{MathStackBinomialIdentity}. We have

\begin{gather*}
\begin{split}
    p^{-\ell} &
    \sum_{j=0}^{n-1} 
    \binom{n+\ell-2-j}{\ell-1}
    \left( 1 - \frac{1}{p}  \right)^{-(j+1)} \\
    &= 
    p^{-\ell} 
    \sum_{j=0}^{n-1} 
    \binom{n+\ell-2-j}{\ell-1} \sum_{k=0}^\infty \binom{j+k}{k} p^{-k} \\
    &= 
    \sum_{k=0}^\infty p^{-(\ell+k)} 
    \sum_{j=0}^{n-1} 
    \binom{n+\ell-2-j}{\ell-1}  \binom{j+k}{k} \\
    &= 
    \sum_{k=0}^\infty  p^{-(\ell+k)}  
    \binom{n+\ell+k-1}{\ell+k}   \\
    &= 
    \sum_{j=0}^\infty   
    \binom{n+j-1}{j} p^{-j} - 
    \sum_{j=0}^{\ell-1}   
    \binom{n+j-1}{j} p^{-j}  \\
    &= 
    \left( 1 - \frac{1}{p}  \right)^{-n} - \sum_{j=0}^{\ell-1}   
    \binom{n+j-1}{j} p^{-j} \,.
\end{split}
\end{gather*}

\end{proof}

\begin{prop}  \label{horrible_identity}
For a dimension $d= \prod_{i \in \mathcal J} p_{i}^{\ell_i}$, with $\{p_{i \in \mathcal J}\}$ the prime factors of $d$, and $m \in \mathbb{Z}_d$ such that $\gcd(m,d)=\prod_{i \in \mathcal J} p_{i}^{k_i}$, it holds that
\begin{equation} 
   \frac{1}{d^n} 
  \sum_{y=0}^{d-1}  \sum_{t=0}^{d-1} \omega^{m \, ty}  \vert \zeta_{y} (n-1) \vert =
   \prod_{i \in \mathcal J} 
    \left[ 1- (1- \delta_{k_i, \ell_i} ) \left( 1 - \frac{1}{p_i} \right)^{n}
    \sum_{j=0}^{\ell_i-k_i-1} p_i^{-j} \binom{n+j-1}{j}
    \right] \,.
\end{equation}
\end{prop}
\begin{proof}
We use the property of the greatest common divisor $\gcd(ca,cb) = c \gcd (a,b)$ (for $c$ a non-negative integer), as well as a few identities and previous results, namely Eq.\,\ref{delta}, Eq.\,\ref{gcd_mod}, and Prop.\,\ref{pi_property}.  
The introduction of Euler's totient function follows the same reasoning as in Prop.\,\ref{prop_x_recursive}. Notice that in the following we use again the notation $\mathcal{D}(r)$ for the set of divisors of an integer $r$.
\begin{gather}
\begin{split}
    \frac{1}{d^n} 
    \sum_{y=0}^{d-1}  \sum_{t=0}^{d-1} \omega^{m \,t \, y}  \vert \zeta_{y} (n-1) \vert 
    &= 
    \frac{1}{d^{n-1}} \sum_{y = 0}^{d-1} \delta_{0, \, m y \bmod d} \vert \zeta_{y} (n-1) \vert \\
    &= 
    \frac{1}{d^{n-1}} \sum_{y = 1}^{d} \delta_{0, \, m y \bmod d} \vert \zeta_{y} (n-1) \vert \\
    &= \frac{1}{d^{n-1}} \sum_{y=1}^{\gcd (m,d)} \left| \zeta_{y \, d / \gcd ( m, d)} (n-1) \right| \\
    &= \frac{1}{d^{n-1}} \sum_{y=1}^{\gcd (m,d)} \left| \zeta_{ \gcd \left( y \, d / \gcd ( m, d), \, d \right)} (n-1) \right| \\
    &= \frac{1}{d^{n-1}} \sum_{y=1}^{\gcd (m,d)} \left| \zeta_{ \gcd \left( y , d \right) \, d / \gcd ( m, d)}  (n-1) \right| \\
    &= \frac{1}{d^{n-1}} 
    \sum_{y \in \mathcal{D} (\gcd(m,d)) } \varphi \left(  \gcd ( m,d) / y \right) \, \left| \zeta_{y \, d / \gcd(m,d) } (n-1) \right| \,.
\end{split}
\end{gather}
Using the prime factorizations of $d$, $\gcd(m,d)$ and $y$, and factorizing the cardinalities of string-sets as in Eq.\,\ref{pi_factorization}, the previous equation becomes 
\begin{gather} \label{intermediate_app_b}
\begin{split}
    \frac{1}{d^n} 
  \sum_{y=0}^{d-1}  \sum_{t=0}^{d-1} \omega^{m \, t \, y}  \vert \zeta_{y} (n-1) \vert &= 
    \frac{1}{d^{n-1}} 
    \sum_{j_1=0}^{k_1}  \sum_{j_2=0}^{k_2} \dots \sum_{j_{| \mathcal J | }=0}^{k_{| \mathcal J | }}
    \varphi \left( \prod_{i \in \mathcal J} p_i^{k_i-j_i} \right)
    \left| \zeta_{\prod_{i\in \mathcal J} p_i^{j_i + \ell_i - k_i} } (n-1) \right|  \\
    &= 
    \frac{|\zeta_1 (n-1)|}{d^{n-1}} 
    \sum_{j_1=0}^{k_1} \sum_{j_2=0}^{k_2} \dots \sum_{j_{| \mathcal J | }=0}^{k_{| \mathcal J | }}
    \prod_{i \in \mathcal J}
    \frac{\varphi (p_i^{k_i-j_i} ) }{|\zeta_1 (n-1)|}
    \left| \zeta_{ p_i^{j_i + \ell_i - k_i} } (n-1) \right|  \\
    &= 
    \frac{ |\zeta_1 (n-1)| }{d^{n-1}} 
    \prod_{i \in \mathcal J}  \sum_{j=0}^{k_i}
    \frac{\varphi (p_i^{k_i-j} ) }{|\zeta_1 (n-1)|}
    \left| \zeta_{ p_i^{j + \ell_i - k_i} } (n-1) \right| \,.
\end{split}
\end{gather}

Let us focus on the sum. We express the cardinalities of string-sets explicitly, Theor.\,\ref{thm_card_generic}, and use the identity in Lemma\,\ref{lem_identity}.

\begin{gather}
\begin{split}
    &\sum_{j=0}^{k_i}
    \frac{\varphi (p_i^{k_i - j})}{ | \zeta_1 (n) | }    
    \left| \zeta_{p_i^{\ell_i-k_i+j}} (n) \right|  \\  
    &= (1-\delta_{k_i,0})  
    \sum_{j=0}^{k_i-1} 
    \frac{\varphi (p_i^{k_i - j})}{ | \zeta_1 (n) | }
    \left| \zeta_{p_i^{\ell_i-k_i-j}} (n) \right| + \frac{\left| \zeta_{p_i^{\ell_i}} (n) \right|}{| \zeta_1 (n) |} \\
    &= (1-\delta_{k_i,0}) \, \varphi(p_i^{k_i}) \sum_{j=0}^{k_i-1} 
    p_i^{-j} \binom{n+\ell_i-k_i+j-1}{\ell_i-k_i+j} 
    + \sum_{j=0}^{n-1} \binom{n+\ell_i-j-2}{\ell_i-1}
    \left( 1 - \frac{1}{p_i} \right)^{-j} \\
    &= (1-\delta_{k_i,0}) \, \varphi(p_i^{\ell_i}) \sum_{j=0}^{k_i-1} 
    p_i^{-(\ell_i-k_i+j)} \binom{n+\ell_i-k_i+j-1}{\ell_i-k_i+j} \\
    &+ \left( 1 - \frac{1}{p_i} \right)^{-n+1} p_i^{\ell_i}
    - \left( 1 - \frac{1}{p_i} \right) p_i^{\ell_i} \, 
    \sum_{j=0}^{\ell_i-1} p_i^{-j} \binom{n+j-1}{j} \\ 
    &= p_i^{\ell_i} \left( 1 - \frac{1}{p_i} \right)
    \left[ 
    \left( 1 - \frac{1}{p_i} \right)^{-n}  +  
    (1-\delta_{k_i,0}) \sum_{j=\ell_i-k_i}^{\ell_i-1} 
    p_i^{-j} \binom{n+j-1}{j} -
    \sum_{j=0}^{\ell_i-1} p_i^{-j} \binom{n+j-1}{j}
    \right] \\ 
    &= p_i^{\ell_i} \left( 1 - \frac{1}{p_i} \right)
    \left[ 
    \left( 1 - \frac{1}{p_i} \right)^{-n}  -  
    (1-\delta_{k_i, \ell_i}) 
    \sum_{j=0}^{\ell_i-k_i-1} p_i^{-j} \binom{n+j-1}{j}
    \right] \,.
\end{split}
\end{gather}

Substituting this back into Eq.\,\ref{intermediate_app_b}, and calculating $|\zeta_1 (n-1)|$ (see Eq.\,\ref{pi_coprime}), proves the proposition.

\end{proof}

\def\thesection{Appendix \Alph{section}}  
\section{Properties of entanglement of elementary hypergraph states}
\renewcommand{\thesection}{\Alph{section}}

We prove two properties of multipartite entanglement of elementary hypergraph states. The first one is for given dimension and multiplicity of the hyperedge; the second one is for given dimension and number of qudits.

\begin{prop} \label{ent_mono}
For given dimension $d$ and multiplicity  of the hyperedge $m_e$, the multipartite entanglement of elementary hypergraph states is strictly monotonically decreasing in the number of qudits $n$, or
\begin{equation}
    E(| G^{m_e}_{n+1} \rangle) <  E(| G^{m_e}_{n} \rangle) \,.
\end{equation}
\end{prop}
\begin{proof}
Proving the proposition, amounts to proving that for every $k_i < \ell_i$ (by construction there has to be at least one), 
\begin{equation}
    \left( 1 - \frac{1}{p_i} \right)^{n} \,
    \sum_{j=0}^{\ell_i-k_i-1} p_i^{-j} \binom{n+j-1}{j}
    < 
    \left( 1 - \frac{1}{p_i} \right)^{n-1} \,
    \sum_{j=0}^{\ell_i-k_i-1} p_i^{-j} \binom{n+j-2}{j} \,.
\end{equation}
For $k_i=\ell_i-1$ the proof is trivial. Let us consider the case $k_i<\ell_i-1$. Using the recurrence relation of binomial coefficients, we have
\begin{equation}
\begin{aligned}
    \left( 1 - \frac{1}{p_i} \right) \,
    \sum_{j=0}^{\ell_i-k_i-1} p_i^{-j} \binom{n+j-1}{j}
    &< 
    \sum_{j=0}^{\ell_i-k_i-1} p_i^{-j} \binom{n+j-2}{j} \\
    \sum_{j=0}^{\ell_i-k_i-1} p_i^{-j} \binom{n+j-1}{j}  -\sum_{j=0}^{\ell_i-k_i-1} p_i^{-j} \binom{n+j-2}{j}
    &< 
    \sum_{j=0}^{\ell_i-k_i-1} p_i^{-j+1} \binom{n+j-1}{j}
     \\
    \sum_{j=1}^{\ell_i-k_i-1} p_i^{-j} \binom{n+j-2}{j-1}
    &< 
    \sum_{j=0}^{\ell_i-k_i-1} p_i^{-j+1} \binom{n+j-1}{j}
     \\
    \sum_{j=0}^{\ell_i-k_i-2} p_i^{-j+1} \binom{n+j-1}{j}
    &< 
    \sum_{j=0}^{\ell_i-k_i-1} p_i^{-j+1} \binom{n+j-1}{j} \,.
     \\
\end{aligned}
\end{equation}
\end{proof}

\begin{prop} \label{ent_multi}
For given dimension $d$ and number of qudits $n$, among the elementary hypergraph states having hyperedge multiplicity multiple of a same divisor $m^*$ of $d$, the one having multiplicity equal to $d/\operatorname{lpf}(m^*)$ is minimally entangled, or
\begin{equation}
    \min_{1 \le \beta < d/m^*} \left( E(| G^{\beta  m^*}_{n} \rangle) \right) =  E(| G^{d/\operatorname{lpf}(d / m^*)}_{n} \rangle) \,.
\end{equation}
\end{prop}
\begin{proof}
Using the prime factors of $d$, $\{ p_{i \in \mathcal J} \}$, we factorize $m^*$ and $d$ as
\begin{equation}
    m^*=\prod_{i \in \mathcal J} p_i^{k_i} \,, \quad
    d = \prod_{i \in \mathcal J} p_i^{\ell_i} \,.
\end{equation}
For Theor.\,\ref{elementary_entanglement}, the states $| G^{\beta m^*}_{n} \rangle$ such that $\beta$ is coprime to $d$ have the same entanglement as $| G^{m^*}_{n} \rangle$; those such that $\beta$ is not coprime to $d$ have lower entanglement. Among these latter, a minimally entangled one has to be such that $\gcd(\beta m^*, d)=d/p_i$, with $p_i$ a prime factor of $d$. In this way, the only element which is not one in the product in Eq.\,\ref{entanglement_general} is that for the index $i$. Among the $p_i$'s, the one which minimizes entanglement, is the smallest one such that $k_i<\ell_i$ in the factorization above, namely $p_i = \operatorname{lpf} (d / m^*)$. This corresponds to a $\beta = d / \left( m^* \, \operatorname{lpf} (d / m^*) \right) $.
\end{proof}

\section*{References}

\bibliographystyle{unsrt.bst}
\bibliography{bibliography}

\end{document}